\documentclass[journal]{IEEEtran}
\ifCLASSINFOpdf
\else
\fi
\usepackage{caption}
\usepackage{cite}
\usepackage{color}
\usepackage{amsmath}
\usepackage{enumerate}%\begin{enumerate}[i] % reacts to 1, a, A, i, and I
\usepackage{enumitem}
\usepackage{mathrsfs}
\usepackage{makeidx}         % allows index generation
\usepackage{graphicx}        % standard LaTeX graphics tool
\usepackage{multicol}        % used for the two-column index
\usepackage{subfigure}
\usepackage{epsfig,amssymb,latexsym}
\usepackage{psfrag}
\usepackage{float}
\usepackage{stfloats}
\usepackage{epsfig,amssymb,latexsym}
\usepackage{psfrag}
\usepackage[ruled,vlined]{algorithm2e}
\usepackage{threeparttable}    %这行要添加
\usepackage{algorithm2e}
\usepackage{caption}
\usepackage{subfigure}
\usepackage{multirow}
\usepackage{amsthm}

\newtheorem{assumption}{Assumption}%[section]
\newtheorem{remark}{Remark}%[section]
\newtheorem{theorem}{Theorem}%[section]
\newtheorem{definition}{Definition}%[section]
\newtheorem{lemma}{Lemma}%[section]
\newtheorem{property}{Property}
\newcommand{\bi}{\begin{itemize}}
	\newcommand{\ei}{\end{itemize}}
\newcommand{\be}{\begin{equation}}
	\newcommand{\ee}{\end{equation}}
\newcommand{\bd}{\begin{displaymath}}
	\newcommand{\ed}{\end{displaymath}}
\newcommand{\bea}{\begin{eqnarray}}
	\newcommand{\eea}{\end{eqnarray}}
\newcommand{\ba}{\begin{array}}
	\newcommand{\ea}{\end{array}}
\newcommand{\bc}{\begin{center}}
	\newcommand{\ec}{\end{center}}

\usepackage{graphicx}
\usepackage{booktabs}

\begin{document}

\title{Real-Time Progressive Learning: Accumulate Knowledge from Control with Neural-Network-Based Selective Memory}

\author{Yiming Fei, Jiangang Li$^{*}$,~\IEEEmembership{Senior Member,~IEEE}, Yanan Li$^{*}$,~\IEEEmembership{Senior Member,~IEEE}

\thanks{Y. Fei was with the School of Mechanical Engineering and Automation, Harbin Institute of Technology, Shenzhen 518055, China and is with the College of Computer Science and Technology, Zhejiang University, Hangzhou 310027, China. Email: yimingfei@zju.edu.cn}
\thanks{*Correspondence: J. Li is with the School of Mechanical Engineering and Automation, Harbin Institute of Technology, Shenzhen 518055, China. Email: jiangang\_lee@hit.edu.cn}
\thanks{*Correspondence: Y. Li is with the Department of Engineering and Design, University of Sussex, Brighton BN1 9RH, UK. Email: yl557@sussex.ac.uk}
}


%这里要改！！！！！！！！！！！！！！！！！


\maketitle

\begin{abstract}
Memory, as the basis of learning, determines the storage, update and forgetting of knowledge and further determines the efficiency of learning. Featured with the mechanism of memory, a radial basis function neural network based learning control scheme named real-time progressive learning (RTPL) is proposed to learn the unknown dynamics of the system with guaranteed stability and closed-loop performance. Instead of the Lyapunov-based weight update law of conventional neural network learning control (NNLC), which mainly concentrates on stability and control performance, RTPL employs the selective memory recursive least squares (SMRLS) algorithm to update the weights of the neural network and achieves the following merits: 1) improved learning speed without filtering, 2) robustness to hyperparameter setting of neural networks, 3) good generalization ability, i.e., reuse of learned knowledge in different tasks, and 4) guaranteed learning performance under parameter perturbation. Moreover, RTPL realizes continuous accumulation of knowledge as a result of its reasonably allocated memory while NNLC may gradually forget knowledge that it has learned. Corresponding theoretical analysis and simulation studies demonstrate the effectiveness of RTPL.
\end{abstract}

%\def\abstractname{Note to Practitioners}
%\begin{abstract}
%	This paper aims to propose an effective five axis contour error control scheme. At present, most of the five axis contour error control methods rely on the modification of the controller, which is not allowed in commercial CNC systems. Therefore, we propose a five axis contour error compensation algorithm based on sILC, which modifies the system's reference path. Its feasibility and advantages are verified by experiments.
%\end{abstract}


\IEEEpeerreviewmaketitle

\section{Introduction}
\IEEEPARstart{T}he control problem of nonlinear systems with uncertainties including unmodeled dynamics, exogenous disturbance, and parameter perturbation has attracted considerable attention in the past few decades. One of the most effective methods to deal with the uncertainties is to collect and utilize the information obtained from the control process and transform the unknown uncertainties into available knowledge, i.e., learning from control. Based on the universal approximation capability of neural networks, methods collectively known as adaptive neural control (ANC) were proposed by using neural networks to approximate uncertainties in real time \cite{r1,r2,r3}. As a result of its strong interpretability and guaranteed stability, ANC has attracted extensive attention in recent years and been applied to various uncertain systems \cite{r4,r5,r6}. However, the learning (weight convergence) of neural networks does not necessarily happen in ANC, which means that neural networks do not take full advantage of their universal approximation capability and have to readjust the weights even when repeating the same control task \cite{r7}. 

Based on the studies of parameter convergence in ANC, a series  of methods collectively named neural network learning control (NNLC) in this paper were proposed to simultaneously identify the uncertainties and complete the control using neural networks \cite{r8,r9,r10,r11}. Compared with traditional ANC, NNLC pays more attention to convergence of weights and aims to learn knowledge about uncertainties from the control process. Generally, NNLC can be divided into two categories according to whether a filter is used to obtain the estimation error based on recorded data. For filter-free NNLC, neural networks update only according to tracking errors measured from the closed-loop system, and both the tracking errors and the estimation errors will exponentially converge with the (partial) persistent exciation (PE) condition \cite{r7,r12}. Compared with ANC, filter-free NNLC achieves real-time learning of uncertainties and thus has attracted extensive attention in recent years. However, one main defect of this method is its slow learning speed, i.e., the neural networks' weights cannot converge accurately enough in a single round of a control task. Therefore, its application is mostly restricted to repetitive tasks \cite{r13,r14,r15,r16}. 

Composite adaptive control combines the features of direct and indirect adaptive control to improve both the control performance and the convergence of parameters \cite{r17,r18,r19}. Inspired by composite adaptive control, filter-based NNLC improves the learning speed of neural networks by constructing a weight update law which considers both the tracking error and the estimation error \cite{r10,r11,r20,r21}. Based on the composite update law, a more practical interval excitation condition is proposed to replace the PE condition to guarantee the exponential convergence of the weights. However, specific filters such as the fixed-point smoother are required to obtain the derivative of the state variables, resulting in increasing computational complexity \cite{r22}. Moreover, the effectiveness of the filter-based NNLC relies heavily on accurate estimation of the optimal weights, but parameter perturbation may cause a mismatch between the recorded data and current dynamics of the system such that the estimation of current optimal weights becomes inaccurate. Since the use of filters and estimation error is impractical in some cases, this paper aims to improve the learning speed of NNLC without filtering. 

Existing studies indicate that the PE condition and PE levels restrict the learning speed of the filter-free NNLC, and thus related research mainly focuses on how to guarantee appropriate PE levels or how to relax the PE condition to improve the learning speed \cite{r11,r23,r24}. Different from the Lyapunov-based approach, this paper analyzes the learning performance of NNLC from the perspective of memory of neural networks. Since the Lyapunov-based design of filter-free NNLC usually leads to a stochastic gradient descent (SGD) based weight update law, the neural network is confronted with the passive knowledge forgetting phenomenon found in \cite{r25}, i.e., knowledge will be gradually forgotten in the learning process. This uncontrolled forgetting is considered to be the underlying cause of the slow learning speed of the filter-free NNLC. 

To suppress passive knowledge forgetting, an effective memory mechanism is required to inform the neural network which data should be learned in the learning process. Therefore, the selective memory recursive least squares (SMRLS) algorithm becomes a potential method instead of SGD to train the radial basis function neural network (RBFNN) \cite{r25}. With SMRLS, the feature space of the RBFNN is uniformly discretized into disjoint partitions and the training samples within the same partition are synthesized into one sample according to memory update functions. Thanks to its unique memory mechanism, SMRLS overcomes passive knowledge forgetting and achieves improved learning speed and generalization capability of the learned knowledge. 

However, SMRLS was originally proposed to address the problem of online supervised learning. For learning control in closed-loop systems, the desired output of neural networks cannot be obtained, so it is impractical to apply SMRLS directly. Inspired by iterative learning control, the desired output of the RBFNN is estimated according to measurable tracking errors \cite{r25_1, r25_2, r25_3}. Based on this idea, a novel learning control scheme named real-time progressive learning (RTPL) is established by using SMRLS to train the RBFNN according to measurable signals in closed-loop systems in this paper. Theoretical analysis shows that the control and learning performance of RTPL can be guaranteed with suitably designed control gains and hyperparameters of the neural network. 

Compared with conventional NNLC, RTPL achieves the following merits: 
\begin{enumerate}
	\item{The learning speed is improved without using filters such that the learning will happen in both repetitive and nonrepetitive control tasks. 
	}
	\item{The robustness to hyperparameter setting of the neural network is improved. 
	 }
	\item{The generalization ability of the learned knowledge is improved as a result of the memory mechanism. 
	}
	\item{The learning performance under parameter perturbation is guaranteed. 
	 }
\end{enumerate}

More interestingly, RTPL achieves the accumulation of knowledge in long-period learning control tasks. Conventional online training methods for the RBFNN such as the SGD and forgetting factor recursive least squares (FFRLS) methods usually include a forgetting mechanism to ensure that the neural network is sensitive to new data \cite{r26,r27,r28}. Therefore, even if the past training samples contain significant knowledge of the system, they will be forgot as a result of the forgetting mechanism. By contrast, RTPL can distinguish which data should be learned and which data should be forgotten and gradually accumulate knowledge with its memory mechanism. 

The rest of this paper is organized as follows. Section \ref{section2} formulates the NNLC problem. Section \ref{section3} introduces preliminaries needed to develop the following theory and approach. Section \ref{section4} designs an SGD based filter-free NNLC scheme and analyzes its passive knowledge forgetting problem. In Section \ref{section5}, the RTPL control scheme is proposed and its stability and parameter convergence are analyzed. Section \ref{section6} provides corresponding simulation results to demonstrate the merits of RTPL. Finally, the main results of this paper are summarized in Section \ref{section7}. 

\section{Problem Formulation}\label{section2}
Consider the following $n$th-order affine nonlinear system in the Brunovsky canonical form   
\begin{equation}\label{eq1}
	\left\{ \begin{aligned}
		{{\dot x}_i} &= {x_{i + 1}},\quad i = 1,2, \ldots ,n \hfill \\
		{{\dot x}_n} &= f(x) + g(x)u \hfill \\ 
	\end{aligned}  \right.
\end{equation}
where $x = {\left[ {{x_1},{x_2}, \ldots ,{x_n}} \right]^T} \in {\mathbb{R}^n}$ is the state vector, $u\in{\mathbb{R}}$ is the continuous system input and  $f(x),g(x):{\mathbb{R}^n} \to \mathbb{R}$ are unknown nonlinear functions.  

\begin{assumption}\label{assumption1}
	Both $f(x)$ and $g(x)$ are unkown ${\mathcal{C}^1}$ functions of $x$. In addition, there exist constants $g_{0l},g_{0u},g_{1l},g_{1u}$ such that $0 < {g_{0l}} \leqslant g(x) \leqslant {g_{0u}}$ and $0<{g_{1l}} \leqslant |\dot{g}(x)| \leqslant {g_{1u}}$.
\end{assumption}

The following reference model is considered  
\begin{equation}\label{eq2}
	\left\{ \begin{gathered}
		{{\dot x}_{di}} = {x_{d(i + 1)}},\quad i=1,2, \ldots ,n \hfill \\
		{{\dot x}_{dn}} = {f_d}({x_d}) \hfill \\ 
	\end{gathered}  \right.
\end{equation}
where ${x_d} = {\left[ {{x_{d1}},{x_{d2}}, \ldots ,{x_{dn}}} \right]^T} \in {\mathbb{R}^n}$, and ${f_d}({x_d}):{\mathbb{R}^n} \to \mathbb{R}$ is a known smooth nonlinear function. It is assumed that $x_d(t)$ is uniformly bounded and satisfies ${x_d}(t) \in {\Omega _d},\forall t \geqslant 0$, where ${\Omega _d} \subset {\mathbb{R}^n}$ is a compact set. The control objective is to design the input $u$ such that the state vector $x$ can track ${x_d}$. 

A classical backstepping method is adopted to design the controller. Consider the following definition of errors  
\begin{equation}\label{eq3}
	\left\{ \begin{aligned}
		{e_1} &= {x_{d1}} - {x_1} \hfill \\
		{e_i} &= {\alpha _{i-1}} - {x_i} ,\quad i = 2, \ldots ,n \hfill \\ 
	\end{aligned}  \right.
\end{equation}
where ${\alpha _{i-1}},i = 2, \ldots ,n$ are virtual control inputs defined as ${\alpha _1} = {k_1}{e_1} + {x_{d2}}$ and ${\alpha _i} =  {{\dot \alpha }_{i - 1}} +{k_i}{e_i} + {e_{i-1}},i = 2, \ldots ,n-1$ with ${k_i>0},i = 1,2, \ldots ,n-1$ being the control gains. The control input $u$ is designed as follows  
\begin{equation}\label{eq4}
	u = {k_n}{e_n} + {e_{n - 1}} + p(x,\dot{\alpha}_{n - 1})
\end{equation}
where $p(x,\dot{\alpha}_{n - 1}) =  {g^{ - 1}}(x)\left( { {{\dot \alpha }_{n - 1}} - f(x) } \right)$. 

According to \cite{r29}, $\dot{\alpha}_{n - 1}$ is a ${\mathcal{C}^1}$ function of $x$ and $x_d$, so $p(x,\dot{\alpha}_{n - 1})$ is represented by $p(x,{x_d})$ in this paper. Since $p(x,{x_d})$ is unknown, it is approximated with neural networks in real time. The objective of the NNLC is to design an effective weight update law such that the learning of $p(x,{x_d})$ and control can be accomplished simultaneously. 

\begin{lemma}\label{lemma1}
	There exists a positive constant ${K_p}$ such that 
\begin{equation}\label{eq5}
\left| {p({x_d}) - p(x,{x_d})} \right| \leqslant {K_p}\left\| e \right\|,\forall x \in {\Omega _x}
\end{equation}
holds for any given compact set ${\Omega _x}$ satisfying ${\Omega _d} \subseteq {\Omega _x}$, where $p({x_d}) = p(x,{x_d}){|_{x = {x_d}}}$ and $e = {\left[ {{e_1},{e_2}, \ldots ,{e_{n - 1}},{e_n}/g(x)} \right]^T}$. 
\begin{proof}
As analyzed in \cite{r30}, $e_i$ can be expressed as follows  
\begin{equation}\label{eq6}
\left\{ \begin{aligned}
	{e_1} &= {x_{d1}} - {x_1} \hfill \\
	{e_i} &= \sum\limits_{j = 0}^{i - 1} {{a_{ij}}e_1^{(j)}} ,\quad i = 2,3, \ldots ,n \hfill \\ 
\end{aligned}  \right.
\end{equation}
where $a_{ij}$ are constant coefficients satisfying $a_{ij}=1$ for $j=i-1$. Transform \eqref{eq6} into the matrix form and we have  
\begin{equation}\label{eq7}
e = {\Lambda _e}({x_d} - x)
\end{equation}
where ${\Lambda _e}$ is a nonsingular lower triangular matrix. From \eqref{eq7} we can obtain  
\begin{equation}\label{eq8}
\begin{aligned}
	\left\| {{x_d} - x} \right\| = \left\| {\Lambda _e^{ - 1}e} \right\| \leqslant {\sigma _{\max }}(\Lambda _e^{ - 1})\left\| e \right\| \hfill \\ 
\end{aligned} 
\end{equation}
in which ${\sigma _{\max }}( \cdot )$ represents the maximal singular value. 

With Assumption \ref{assumption1} and $\dot{\alpha}_{n - 1}$ being a $\mathcal{C}^1$ function of $x$ and $x_d$, $p(x,x_d)$ is a ${\mathcal{C}^1}$ function of $x$ and $x_d$ defined on ${\Omega _x} \times {\Omega _d}$. Therefore, $p(x,x_d)$ is Lipschitz on ${\Omega _x} \times {\Omega _d}$ which means that there exists a positive constant $K_0$ such that  
\begin{equation}\label{eq9}
\left| {p({x_d}) - p(x,{x_d})} \right| \leqslant {K_0}\left\| {{x_d} - x} \right\|,\forall x \in {\Omega _x}. 
\end{equation}
Substitute \eqref{eq8} into \eqref{eq9} and the inequation \eqref{eq5} is derived where ${K_p} = {K_0}{\sigma _{\max }}(\Lambda _e^{ - 1})$. 
\end{proof}
\end{lemma}

\begin{remark}\label{remark1}
	The reason why the NNLC of system \eqref{eq1} is considered is to facilitate the understanding of RTPL in subsequent sections. The proposed RTPL method also has the potential to be applied to other systems such as strict-feedback systems and nonaffine systems \cite{r13,r31}. However, this is not within the scope of this paper. 
\end{remark}

\section{Preliminaries}\label{section3}
\subsection{RBFNN and Its Approximation Capability}\label{section2-2}
The neural network used to approximate unknown functions in this paper is the linearly parameterized RBFNN whose output can be formulated as follows   
\begin{equation}\label{eq10}
{f_{NN}}(\chi ) = \sum\limits_{i = 1}^N {{w_i}{\phi _i}(\chi ) = } {W^T}\Phi (\chi )
\end{equation}
where $N$ is the number of neurons in the hidden layer, $\chi  \in {\mathbb{R}^q}$ is the input vector, ${f_{NN}}(\chi ) \in \mathbb{R}$ is the scalar output, $W = {[{w_1},{w_2}, \ldots ,{w_N}]^T} \in {\mathbb{R}^N}$ is the weight vector, and $\Phi (\chi ) = {\left[ {{\phi _1}(\chi ),{\phi _2}(\chi ), \ldots ,{\phi _N}(\chi )} \right]^T} \in {\mathbb{R}^N}$ is the regressor vector composed of radial basis functions ${\phi _i}(\chi ),i = 1,2, \ldots ,N$. 

The Gaussian RBFNN is adopted in this paper whose radial basis function is formulated as  
\begin{equation}\label{eq11}
{\phi _i}(\chi ) = \exp \left( { - {{\left\| {\chi  - {c_i}} \right\|}^2}/2\sigma _i^2} \right),i = 1,2, \ldots N
\end{equation}
where ${c_i} \in {\mathbb{R}^q}$ is the center and ${\sigma _i} \in \mathbb{R}$ is the receptive field width of ${\phi _i}(\chi )$. 

\begin{assumption}\label{assumption2}
	In this paper, the lattice distribution of the neuron centers $c_i$ is adopted, which covers the region ${\Omega _d}$. 
\end{assumption}

\begin{property}\label{property1}
(Universal Approximation \cite{r32}): RBFNNs have the capability to approximate any continuous function $h(\chi):{\Omega _\chi} \to \mathbb{R}$ over a compact set ${\Omega _\chi } \subset {\mathbb{R}^q}$ to arbitrary precision on the premise of enough neurons and suitable centers and receptive field widths. Then $h(\chi)$ can be approximated as  
\begin{equation}\label{eq12}
h(\chi ) = {W^{*T}}\Phi (\chi ) + \epsilon(\chi ),\forall \chi  \in {\Omega _\chi }
\end{equation}
where $W^{*}\in {\mathbb{R}^N}$ is the optimal weight vector, and $\epsilon (\chi)$ is a bounded approximation error satisfying $\left| {\epsilon(\chi )} \right| < {\epsilon^*},\forall \chi  \in {\Omega _\chi }$. In this paper, $W^{*}$ is defined as follows  
\begin{equation}\label{eq13}
{W^*} \triangleq \mathop {\arg \min }\limits_{W \in {\mathbb{R}^N}} \left\{ {{{\int_{{\Omega _\chi }} {\left\| {h(\chi ) - {W^T}\Phi (\chi )} \right\|} }^2}d\chi } \right\}.
\end{equation}	
\end{property}

\begin{property}\label{property2}
(Localized Approximation Capability \cite{r7}):  Assume that ${\Omega _{\chi j}}, j = 1,2, \ldots ,M$ are different regions over the compact set ${\Omega _\chi }$, ${W_j}$ is the weight vector of the neurons whose centers are close to the region ${\Omega _{\chi j}}, $ $W_j^ *  \in {\mathbb{R}^{{N_j}}}$ is the corresopnding subvector of ${W^ * }$, and the local approximation of $h(\chi )$ over ${\Omega _{\chi j}}$ can be performed with $W_j^ * $ as follows  
\begin{equation}\label{eq14}
\begin{aligned}
	h(\chi ) &= {W^{*T}}\Phi (\chi ) + \epsilon(\chi ) \\ 
	&= W_j^{ * T}{\Phi _j}(\chi ) + {\epsilon_j}(\chi ),\forall \chi  \in {\Omega _{\chi j}} \\ 
\end{aligned} 
\end{equation}
where ${\Phi _j}(\chi )$ is a subvector of $\Phi (\chi )$ corresponding to $W_j^ * $, and $\epsilon{_j}(\chi )$ is a bounded approximation error satisfying $\left| {{\epsilon_j}(\chi )} \right| < \epsilon_j^ * ,\forall \chi  \in {\Omega _{\chi j}}$ with the value of $\left| {\epsilon(\chi ) - {\epsilon_j}(\chi )} \right|$ being small. 
\end{property}

\subsection{Partial Persistent Excitation Condition for RBFNNs}\label{section2-3}
The PE condition of the regressor vector $\Phi (\chi )$ is an important concept in filter-free NNLC which can lead to exponential convergence of the weight estimation error. 
\begin{definition}\label{definition1}
(Persistent Excitation Condition \cite{r28}): Assume that $\Phi (\chi ):{\mathbb{R}^q} \to {\mathbb{R}^N}$ is a piecewise-continuous and uniformly-bounded regressor vector. Then it satisfies the PE condition if there exist positive constants $\delta $, ${\alpha _1}$, ${\alpha _2}$ such that
\begin{equation}\label{eq15}
	{\alpha _1}I \geqslant \int_{{t_0}}^{{t_0} + \delta } {\Phi (\tau ){\Phi ^T}(\tau )d\tau  \geqslant } {\alpha _2}I
\end{equation}
	holds for  $\forall {t_0} \geqslant 0$, where $I \in {\mathbb{R}^{N \times N}}$ is an identity matrix. 
\end{definition}

Let the subscript $\zeta$ denote the neurons whose center is close to the input trajectory of the RBFNN and the PE condition of the regressor subvector ${\Phi _\zeta }(\chi )$ is named as the partial PE condition of $\Phi (\chi )$ \cite{r7}. 

\begin{assumption}\label{assumption3}
The PE condition of the regressor subvector $\Phi_\zeta(\chi)$ is satisfied with any input trajectory of the RBFNN whose duration is finite. 
\end{assumption}

\begin{remark}\label{remark2}
Existing studies indicate that the recurrent input trajectory of the RBFNN is a sufficient condition for the partial PE condition \cite{r7,r33}. However, we argue that the partial PE condition can always be used as a priori knowledge in theoretical analysis, whether the finite-duration input trajectory of the RBFNN is recurrent or not. The rationality of Assumption \ref{assumption3} can be verified through the following steps. As analyzed in \cite{r7,r33}, the regressor subvector ${\Phi _\zeta }(\chi )$ is always sufficiently excitated for any given finite trajectory $\chi (t),t \in \left[ {0,{T_0}} \right]$, which means that the integral $\int_{{0}}^{ T_0 } {\Phi_{\zeta} (\tau ){\Phi ^T_{\zeta}}(\tau )d\tau } $ is nonsingular. Let ${\varphi _\chi }(t),t \in \left[ {0, + \infty } \right)$ denote a periodic infinite trajectory whose initial stage is $\chi (t),t \in \left[ {0,{T_0}} \right]$, and it is easy to obtain that the original subvector ${\Phi _\zeta }(\chi )$, whose neuron centers are close to $\chi (t)$, satisfies the PE condition along ${\varphi _\chi }(t)$. Since ${\varphi _\chi }(t),t \in \left[ {0,{T_0}} \right]$ is the same as $\chi (t),t \in \left[ {0,{T_0}} \right]$, NNLC along these two trajectories shares the same convergence process on the time interval $\left[ {0,{T_0}} \right]$. Therefore, for any given $\chi (t),t \in \left[ {0,{T_0}} \right]$, the partial PE condition can be used to analyze the convergence in NNLC. 
\end{remark}

\subsection{Selective Memory Recursive Least Squares}\label{section2-3}
To overcome the passive knowledge forgetting problem caused by forgetting factors, a recursive least squares (RLS) based real-time training method named SMRLS is proposed to improve the learning performance of the RBFNN \cite{r25}. 

Consider the function approximation task mentioned in Property \ref{property1}, where the unknown function $h(\chi)$ is approximated by an RBFNN over its input space $\Omega_\chi$. With Assumption \ref{assumption2}, the input space normalization is omitted, and $\Omega_\chi$ is evenly discretized into $N_P$ disjoint partitions $\Omega _\chi ^j,j = 1,2, \ldots ,{N_P}$. Let $\chi (k+1)$, $h(k+1)$ denote the value of $\chi$ and $h(\chi)$ at sampling time $k+1$, respectively, and it is assumed that $\chi (k + 1) \in \Omega _\chi ^a$, i.e. the partition $\Omega _\chi ^a$ is being sampled at sampling time $k+1$. 

According to SMRLS \cite{r25}, the weights of the RBFNN can be updated using the following update laws  
\begin{itemize}[leftmargin=*]
	\item \textbf{Selective Memory Recursive Least Squares}
	\begin{equation}\label{eq16}
		\begin{footnotesize}
			\!\! \! \! \! \! \! \!  \begin{aligned}
				&W(k + 1) = \\
		 &W(k) + P(k + 1)\Phi (k + 1)\left[ {h(k + 1) - {W^T}(k)\Phi (k + 1) }  \right] - \epsilon_a(k+1)\\
				\\
				&{P^{ - 1}}(k+1) ={P^{ - 1}}(k)+ \Phi (k + 1) \Phi (k + 1)^T -\Phi_a (k){\Phi_a ^T}(k) 
			\end{aligned} 
		\end{footnotesize}
	\end{equation}
\end{itemize}
where ${\epsilon _a}(k + 1) = P(k + 1){\Phi _a}(k)[{h_a}(k) - {W^T}(k){\Phi _a}(k)]$, $P(k + 1) \in {\mathbb{R}^{N \times N}}$ is a positive definite matrix, ${\Phi_a }(k)$ and ${h_a}(k)$ are the recorded regressor vector and desired neural network output of partition $\Omega _\chi ^a$ before sampling time $k$, respectively. Since ${h_a}(k) - {W^T}(k){\Phi _a}(k)$ represents the network approximation error to a sample that has been learned, it becomes a small value if the approximation capability of the neural network is sufficient. Therefore,  $\left| {{\epsilon _a}(k + 1)} \right| $ will also be small if the neural network is set properly. 

To facilitate theoretical analysis, it is assumed that the sampling period is small enough such that \eqref{eq16} can be reformulated into the continuous form  
\begin{equation}\label{eq17}
\left\{ \begin{gathered}
	\dot W = P\Phi (\chi )\left[ {h - {W^T}\Phi (\chi )} \right] - {\epsilon _a} \hfill \\
	{{\dot P}^{ - 1}} = \Phi {\Phi ^T} - {\Phi _a}\Phi _a^T \hfill \\ 
\end{gathered}  \right.
\end{equation}
where ${\epsilon _a} = P{\Phi _a}({h_a} - {W^T}{\Phi _a})$ satisfying $\left\| {{\epsilon _a}} \right\| < \epsilon _a^*$ is the continuous representation of ${\epsilon _a}(k+1)$. It should be noted that the symbol $\dot{P}^{-1}$ represents the derivative of $P^{-1}$.  

\begin{assumption}\label{assumption4}
In this paper, the approximation capability of the neural network is sufficient such that $\left\| {{\epsilon _a}} \right\| < \epsilon _a^*$ where $\epsilon _a^*>0$ is a small error upper bound determined by hyperparameters of the neural network. 
\end{assumption}

\begin{remark}\label{remark3}
The initial values are set as $P(0) = {p_0}I $, $W(0) = {0_{N \times 1}}$, where ${p_0} > 0$ is a large constant and $I\in {\mathbb{R}^{N \times N}}$ is an identity matrix. In addition, the recorded regressor vector ${\Phi _a}$ and desired output $h_a$ of each partition $\Omega_\chi^a, a = 1,2,\ldots,N_P$ are set to ${0_{N \times 1}}$ and $0$, respectively. 
\end{remark}
Consider the discrete update law of ${P}^{-1}(k)$ in \eqref{eq16}. Assuming that there are $N_k$ partitions that have been sampled at least once before sampling time $k$, and the recorded regressor subvectors of the partitions are ${\Phi(j)},j = 1,2, \ldots ,{N_k}$, then  ${P}^{-1}(k)$ can be calculated by  
\begin{equation}\label{eq18}
{P^{ - 1}}(k) = {P^{ - 1}}(0) + \sum\limits_{j = 1}^{{N_k}} {\Phi (j){\Phi ^T}(j)}.
\end{equation}
According to Remark \ref{remark3}, $P^{-1}(0)=\frac{1}{p_0}I$, so $P(k)$ and $P^{ - 1}(k)$ are always positive definite in the learning process. 

According to Courant-Fischer min-max theorem, we have  
\begin{equation}\label{eq19}
	\begin{small}
	\begin{aligned}
		{\lambda _{\min }}({P^{ - 1}}(k)) &= \mathop {\min }\limits_{X \in {\mathbb{R}^N},\left\| X \right\| \ne 0} \frac{{{X^T}{P^{ - 1}}(k)X}}{{{X^T}X}} \hfill \\
		&= \mathop {\min }\limits_{X \in {\mathbb{R}^N},\left\| X \right\| \ne 0} \frac{{{X^T}\left[ {{P^{ - 1}}(0) + \sum\limits_{j = 1}^{{N_k}} {\Phi (j){\Phi ^T}(j)} } \right]X}}{{{X^T}X}} \hfill \\
		&= {\lambda _{\min }}({P^{ - 1}}(0)) + {\lambda _{\min }}\left( {\sum\limits_{j = 1}^{{N_k}} {\Phi (j){\Phi ^T}(j)} } \right) \hfill \\
		&\geqslant {\lambda _{\min }}({P^{ - 1}}(0)) \hfill \\ 
	\end{aligned} 
	\end{small}
\end{equation}
where $\lambda_{\min}(\cdot)$ represents the minimum eigenvalue. Since ${\lambda _{\max }}(P(k)) = 1/{\lambda _{\min }}({P^{ - 1}}(k))$, we obtain ${\lambda _{\max }}(P(k)) \leqslant {\lambda _{\max }}(P(0))$ holds for any $k>0$, where $\lambda_{\max}(\cdot)$ represents the maximum eigenvalue. Consider the initialization in Remark \ref{remark3}, and the following inequation  
\begin{equation}\label{eq20}
{\lambda _{\max }}(P(k)) \leqslant {p_0}
\end{equation}
holds for any $k>0$. Similarly, the maximum eigenvalue of $P^{-1}(k)$ satisfies  
\begin{equation}\label{eq21}
\begin{aligned}
	{\lambda _{\max }}({P^{ - 1}}(k)) &= {\lambda _{\max }}({P^{ - 1}}(0)) + {\lambda _{\max }}\left( {\sum\limits_{j = 1}^{{N_k}} {\Phi (j){\Phi ^T}(j)} } \right). \hfill \\
\end{aligned} 
\end{equation}
Since there are $N_P$ partitions in total, we have ${N_k} \leqslant {N_P}$ indicating that ${\lambda _{\max }}({P^{ - 1}}(k))$, i.e. $1/{\lambda _{\min }}(P(k))$, will not increase endlessly even for an infinite trajectory of $x_d$. Therefore, there exists a positive constant $q_0>0$ such that ${q_0} \leqslant {\lambda _{\min }}(P(k))$ holds for any $k>0$. 

Assume that the sampling period is small enough such that \eqref{eq16} is equivalent to \eqref{eq17} and the following assumption is considered. 
\begin{assumption}\label{assumption5}
With the initialization in Remark \ref{remark3}, $P$ and $P^{-1}$ are always positive definite along the trajectory of \eqref{eq17}. In addition, there exist positive constants $q_0$ and $p_0$ such that $0<q_0 \leqslant \lambda_{\min}(P) \leqslant {\lambda _{\max }}(P) \leqslant {p_0}$ always holds. 
\end{assumption}

\section{Stochastic Gradient Descent Based Neural Network Learning Control}\label{section4}
Under the partial PE condition, the filter-free NNLC adopts a gradient descent based weight update law and achieves exponential convergence of its weight subvector \cite{r8, r9}. In this section, a filter-free NNLC scheme in the form of an adaptive feedforward controller is designed and its passive knowledge forgetting problem is analyzed. 

\subsection{Adaptive RBFNN Control and Its Learning Phenomenon}\label{section4-1}
Consider the tracking control problem in Section \ref{section2} and the backstepping control law \eqref{eq4}. An RBFNN is employed to approximate the unknown function $p(x_d)$. According to Property \ref{property2}, let the subscript $\zeta$ denote the neurons close to the trajectory of $x_d$ and $p(x_d)$ can be approximated by 
\begin{equation}\label{eq22}
p({x_d}) = {W_\zeta^{ * T}}\Phi_\zeta ({x_d}) + {\epsilon _{p1}}({x_d})
\end{equation}
where ${\epsilon _{p1}}({x_d})$ is a bounded approximation error determined by hyperparameters of the RBFNN and satisfies $\left| {{\epsilon _{p1}}({x_d})} \right| < \epsilon _{p1}^ * $, $W_\zeta^{ *}\in {\mathbb{R}^{N_\zeta}}$ and $\Phi_\zeta ({x_d})\in {\mathbb{R}^{N_\zeta}}$ are corresponding subvectors of $W^*$ and $\Phi(x_d)$, respectively. The optimal weight vector $W^{*}$ is defined as follows  
\begin{equation}\label{eq23}
{W^ * } = \mathop {\arg \min }\limits_{W \in {\mathbb{R}^N}} \left\{ {{{\int_{{\Omega _d}} {\left\| {p({x_d}) - {W^T}\Phi ({x_d})} \right\|} }^2}d{x_d} } \right\}.
\end{equation}

As analyzed in \cite{r34}, a hybrid feedforward feedback control law is established by replacing $p(x,\dot{\alpha}_{n - 1})$ of \eqref{eq4} with $p(x_d)$ as follows  
\begin{equation}\label{eq24}
	u = {k_n}{e_n} + {e_{n - 1}} + p(x_d). 
\end{equation}

Let $\hat W_\zeta \in {\mathbb{R}^{N_\zeta}}$ denote the approximation of $W_\zeta^{*}$. Define the weight approximation error $\tilde W_\zeta = {W_\zeta^ * } - \hat W_\zeta$, the function tracking error ${\epsilon _{p2}}(x,{x_d}) = p({x,x_d}) - p(x_d)$, and the approximation of $p({x_d})$ is formulated as  
\begin{equation}\label{eq25}
\begin{aligned}
	\hat p(x_d) &= {{\hat W}^T}\Phi ({x_d}) \hfill \\
	&= \hat W_\zeta ^T{\Phi _\zeta }({x_d}) - {\epsilon _\zeta }({x_d}) \hfill \\
	&= W_\zeta ^{*T}{\Phi _\zeta }({x_d}) - \tilde W_\zeta ^T{\Phi _\zeta }({x_d}) - {\epsilon _\zeta }({x_d}) \hfill \\
	&= p({x,x_d}) - {{\tilde W}_\zeta^T}\Phi_\zeta ({x_d}) - {\epsilon _p}(x,{x_d}) \hfill \\ 
\end{aligned} 
\end{equation}
where ${\epsilon _\zeta }({x_d})$ satisfying $\left| {{\epsilon _\zeta }({x_d})} \right| < \epsilon _\zeta ^*$ is the approximation error caused by the neurons apart from the neuron group $\zeta$, ${\epsilon _p}(x,{x_d}) =   {\epsilon _{p1}}({x_d}) + {\epsilon _{p2}}(x,{x_d}) + {\epsilon _\zeta }({x_d})$ is the composite approximation error. 

Approximate $p(x_d)$ with the RBFNN and the control law \eqref{eq24} becomes  
\begin{equation}\label{eq26}
	\begin{aligned}
	u &= {k_n}{e_n} + {e_{n - 1}} + {{\hat W}^T}\Phi ({x_d}) \\
&= {k_n}{e_n} + {e_{n - 1}} + p({x,x_d}) - {{\tilde W}_\zeta^T}\Phi_\zeta ({x_d}) - {\epsilon _p}(x,{x_d}). 
	\end{aligned}
\end{equation}

The weight vector $ {\hat W} $ is updated by the gradient descent based method  
\begin{equation}\label{eq27}
\dot {\hat W} =  - \dot {\tilde W} = \Gamma \Phi ({x_d}){e_n}
\end{equation}
where $\Gamma  = {\Gamma ^T} > 0$ is a positive definite diagonal matrix. Similarly, ${{\hat W}_\zeta }$ is updated by 
\begin{equation}\label{eq28}
{{\dot {\hat W}}_\zeta } =  - {{\dot {\tilde W}}_\zeta } = {\Gamma _\zeta }{\Phi _\zeta }({x_d}){e_n}
\end{equation}
where ${\Gamma _\zeta }$ is is a positive definite diagonal matrix composed of corresponding entries of $\Gamma$. 

Consider the system \eqref{eq1}, the reference model \eqref{eq2}, the error definition \eqref{eq3}, the control law \eqref{eq26}, the weight update law \eqref{eq28}, and the closed-loop dynamics is expressed as  
\begin{equation}\label{eq29}
\left\{ \begin{aligned}
	&{{\dot e}_1} =  - {k_1}{e_1} + {e_2} \hfill \\
	&{{\dot e}_i} =  - {k_i}{e_i} + {e_{i + 1}} - {e_{i - 1}}, i=2,3,\ldots,n-1 \hfill \\
	&{{\dot e}_n} =  - g(x)\left( {{k_n}{e_n} + {e_{n - 1}} - \tilde W_\zeta ^T{\Phi _\zeta }({x_d}) - {\epsilon _p}(x,{x_d})} \right) \hfill \\
	&{{\dot {\tilde W}}_\zeta } =  - {\Gamma _\zeta }{\Phi _\zeta }({x_d}){e_n}. \hfill \\ 
\end{aligned}  \right.
\end{equation}

The matrix form of \eqref{eq29} can be formulated as follows  
\begin{equation}\label{eq30}
	\begin{aligned}
\left[ {\begin{array}{*{20}{c}}
		{\dot e} \\ 
		{{{\dot {\tilde W}}_\zeta }} 
\end{array}} \right] =& \left[ {\begin{array}{*{20}{c}}
		{A}&\vline & { bg(x)\Phi _\zeta ^T({x_d})} \\ 
		\hline
		{ - {\Gamma _\zeta }{\Phi _\zeta }({x_d}){b^T}}&\vline & {{0_{{N_\zeta } \times {N_\zeta }}}} 
\end{array}} \right]\left[ {\begin{array}{*{20}{c}}
		e \\ 
		{{{\tilde W}_\zeta }} 
\end{array}} \right] \\
&+ \left[ {\begin{array}{*{20}{c}}
		{ bg(x){\epsilon _p}(x,{x_d})} \\ 
		{{0_{{N_\zeta } \times 1}}} 
\end{array}} \right]
	\end{aligned}
\end{equation}
where $e = \left[ {{e_1},{e_2}, \ldots ,{e_n}} \right] \in {\mathbb{R}^n}$, $b = \left[ {\begin{array}{*{20}{c}}
		{{0_{(n - 1) \times 1}}} \\ 
		1 
\end{array}} \right] \in {\mathbb{R}^n}$, and $A$ can be expressed as  
\begin{equation}\label{eq31}
A = \left[ {\begin{array}{*{20}{c}}
		{ - {k_1}}&1&{}&{}&{} \\ 
		{ - 1}&{ - {k_2}}&1&{}&{} \\ 
		{}&{}& \ddots &{}&{} \\ 
		{}&{}&{ - 1}&{{-k_{n - 1}}}&1 \\ 
		{}&{}&{}&{ - g(x)}&{ - {k_n}g(x)} 
\end{array}} \right] \in {\mathbb{R}^{n \times n}}. 
\end{equation}

\begin{lemma}\label{lemma2}
Consider the nominal system of \eqref{eq30} as follows  
\begin{equation}\label{eq32}
	\begin{aligned}
		\left[ {\begin{array}{*{20}{c}}
				{\dot e} \\ 
				{{{\dot {\tilde W}}_\zeta }} 
		\end{array}} \right] = \left[ {\begin{array}{*{20}{c}}
				{A}&\vline & { bg(x)\Phi _\zeta ^T({x_d})} \\ 
				\hline
				{ - {\Gamma _\zeta }{\Phi _\zeta }({x_d}){b^T}}&\vline & {{0_{{N_\zeta } \times {N_\zeta }}}} 
		\end{array}} \right]\left[ {\begin{array}{*{20}{c}}
				e \\ 
				{{{\tilde W}_\zeta }} 
		\end{array}} \right] 
	\end{aligned}. 
\end{equation}
With Assumption \ref{assumption1}, the PE condition of ${{\Phi _\zeta }({x_d})}$ and a sufficiently large control gain $k_n$, \eqref{eq32} is exponentially stable, which means that the errors $e$ and ${\tilde W}_\zeta$ will exponentially converge to zero. 
\end{lemma}

\begin{proof}
The proof has been given in Theorem 1 of \cite{r35}. 
\end{proof}

\begin{theorem}\label{theorem1}
Consider the system \eqref{eq30}. With Assumption \ref{assumption1}, the PE condition of ${{\Phi _\zeta }({x_d})}$ and properly designed control gains, $e$ and ${\tilde W}_\zeta$ will exponentially converge to small neighborhoods around zero, which can be arbitrarily contracted with the increase of ${k_1},{k_2}, \ldots ,{k_n}$ and ${\Gamma _\zeta }$. 
\end{theorem}

\begin{proof}
According to Lemma \ref{lemma1}, we obtain $\left| {{\epsilon _{p2}}(x,{x_d})} \right| = \left| {p(x,{x_d}) - p({x_d})} \right| \leqslant {K_p}\left\| e \right\|,\forall x \in {\Omega _x}$. Since ${\epsilon _{p2}}(x,{x_d}) \leqslant {K_p}\left\| e \right\|$, $\left| {{\epsilon _{p1}}({x_d})} \right| < \epsilon _{p1}^ * $, $\left| {{\epsilon _\zeta }({x_d})} \right| < \epsilon _\zeta ^*$, the system \eqref{eq30} can be regarded as a perturbed system whose perturbation includes a vanishing term $bg(x) \epsilon_{p2}(x,x_d)$ and a nonvanishing term $bg(x)(\epsilon_{p1}(x_d)+\epsilon_{\zeta}(x_d))$. 

Since Lemma 2 indicates that the nominal part of \eqref{eq30} is exponentially stable with a sufficiently large $k_n$, Lemma 9.1 and 9.2 of \cite{r36} are used to analyze the influence of $\epsilon_p$ to the nominal system \eqref{eq32}. It is shown that, with sufficiently large control gains ${k_1},{k_2}, \ldots ,{k_n}$ and ${\Gamma _\zeta }$, the errors $e$ and ${\tilde W}_\zeta$ will exponentially converge to small neighborhoods around zero. Moreover, the neighborhoods can be arbitrarily contracted by increasing the gains ${k_1},{k_2}, \ldots ,{k_n}$ and ${\Gamma _\zeta }$. 
\end{proof}

After the transient process, the converged weight vector $\hat{W}$ can approximate $p(x_d)$ accurately along the trajectory of $x_d$. The following integral strategy is adopted to obtain the learned knowledge  
\begin{equation}\label{eq33}
\bar W = \frac{1}{{\Delta t}}\int_{{t_0}}^{{t_0} + \Delta t} {\hat W(\tau )} d\tau 
\end{equation}
where $\left[ {{t_0},{t_0} + \Delta t} \right]$ represents a time interval after the transient process. Thus, the unknown function $p(x_d)$ can be approximated as follows  
\begin{equation}\label{eq34}
\hat p({x_d}) = {{\bar W}^T}\Phi ({x_d})
\end{equation}
which can be used to design the feedforward control for the same tracking control task as in the learning phase.

Although the approach mentioned above realizes a basic knowledge learning-reuse framework, there are still some problems to be solved: 
\begin{enumerate}[leftmargin=*]
\item{Good learning performance depends heavily on high-level PE, which requires the high-frequency recurrence of the input signal, so this SGD based learning control only performs well when the reference trajectory $x_d$ is repetitive \cite{r8, r13, r14, r37}. 
 }

\item {As the PE levels are also influenced by receptive field width of the RBFNN, the width has to be designed properly (usually small enough), which possibly weakens the approximation capability of the RBFNN \cite{r38, r39}. 
}
\end{enumerate}

\subsection{Passive Knowledge Forgetting in Closed-Loop Systems}\label{section4-2}
The causes of the aforementioned problems can be attributed to the algorithm's heavy reliance on high-level PE condition. To provide an intuitive explanation for the reliance, the concept of passive knowledge forgetting proposed in \cite{r25} is considered. However, it is important to note that the analysis of passive knowledge forgetting in a closed-loop system cannot be directly performed by analyzing the objective function as in open-loop systems. In this paper, we adopt an illustrative approach to analyze the impact of passive knowledge forgetting on the learning process in closed-loop systems. 

Consider a special case of the control problem mentioned in Section \ref{section2} where ${x_d} = {\left[ {{x_{d1}},{x_{d2}}} \right]^T} \in {\mathbb{R}^2}$. Let $\varphi_d$ denote the trajectory of $x_d$, and a special $\varphi_d$ is considered, which involves two phases as shown in Fig. \ref{fig1}: i) a repetitive phase ${\varphi _{d1}}$ and ii) a new exploration phase ${\varphi _{d2}}$. 

\begin{figure}[htbp]
	\centering
	\includegraphics[scale=0.5]{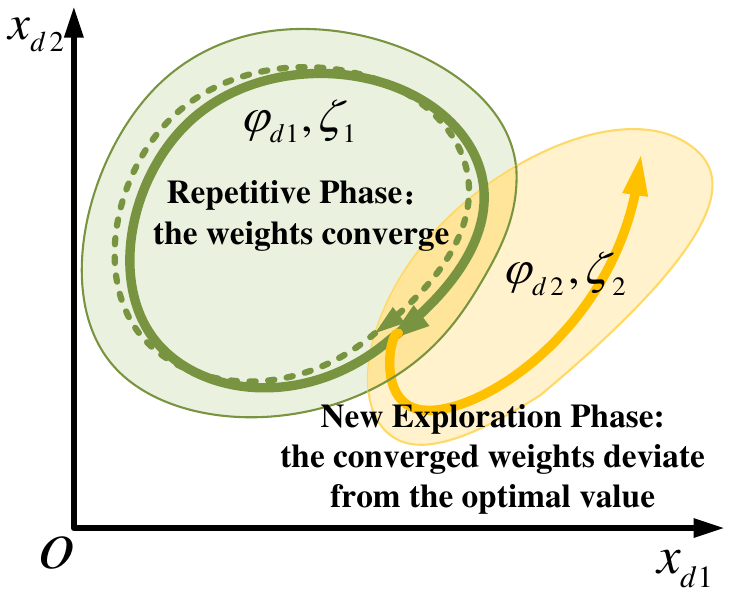}
	\caption{Passive knowledge forgetting phenomenon}\vspace{-1em}\label{fig1}
\end{figure}

According to Property \ref{property2} and Assumption \ref{assumption2}, the approximations of the unknown function $p(x_d)$ along ${\varphi _{d1}}$ and ${\varphi _{d2}}$ are performed by two different groups of neurons which have intersections. Let the subscripts ${\zeta _1}$ and ${\zeta _2}$ represent these two groups of neurons close to the reference trajectories ${\varphi _{d1}}$ and ${\varphi _{d2}}$, respectively, and the unknown function $p(x_d)$ can be approximated as  
\begin{equation}\label{eq35}
\left\{ \begin{gathered}
	p({x_d}) = W_{{\zeta _1}}^{ * T}{\Phi _{{\zeta _1}}}({x_d}) + {\epsilon _{{\zeta _1}}}({x_d}),\forall {x_d} \in {\varphi _{d1}} \hfill \\
	p({x_d}) = W_{{\zeta _2}}^{ * T}{\Phi _{{\zeta _2}}}({x_d}) + {\epsilon _{{\zeta _2}}}({x_d}),\forall {x_d} \in {\varphi _{d2}} \hfill \\ 
\end{gathered}  \right.
\end{equation}
where ${\epsilon _{{\zeta _1}}}({x_d})$ and ${\epsilon _{{\zeta _2}}}({x_d})$ satisfying $\left| {{\epsilon _{{\zeta _1}}}({x_d})} \right| < \epsilon _{{\zeta _1}}^ * $, $\left| {{\epsilon _{{\zeta _2}}}({x_d})} \right| < \epsilon _{{\zeta _2}}^ * $ are bounded approximation errors. In addition, the two shaded areas in Fig. \ref{fig1} represent the approximation domains of neuron groups ${\zeta _1}$ and ${\zeta _2}$, respectively, which means that the local approximation of $p(x_d)$ over these regions will be influenced by ${{\hat W}_{{\zeta _1}}}$ or ${{\hat W}_{{\zeta _2}}}$. 

Consider the SGD based NNLC proposed in Section \ref{section4-1}. According to Theorem \ref{theorem1}, the weight subvector ${{\hat W}_{{\zeta _1}}}$ will converge to a small neighborhood around $W_{{\zeta _1}}^ * $ in the repetitive phase. Thus, the locally accurate knowledge about $p(x_d)$ along the trajectory ${\varphi _{d1}}$ is learned by the RBFNN. However, when the repetitive phase ends and the reference trajectory ${\varphi _d}$ enters the new exploration phase ${\varphi _{d2}}$, the weights of the neuron group ${\zeta _2}$, which has intersections with ${\zeta _1}$, will have the tendency to converge to $W_{{\zeta _2}}^ * $. Therefore, some entries of ${{\hat W}_{{\zeta _1}}}$, which also take part in the approximation of $p(x_d)$ along ${\varphi _{d2}}$, will gradually deviate from the previous optimal value $W_{{\zeta _1}}^ * $, and the learned knowledge in the repetitive phase will be gradually forgotten in the subsequent learning control process. 

It should be noted that the passive knowledge forgetting phenomenon happens not only in the special case illustrated in Fig. \ref{fig1}, but at any time in arbitrary tracking control tasks. Specifically, according to Property \ref{property2}, when the reference trajectory ${\varphi _d}$ visits a small region ${\Omega _{dj}},j = 1,2, \ldots ,M$ of the reference state space $\Omega_d $, the corresponding weight subvector ${{\hat W}_j}$ will have the tendency to converge to its optimal value $W_j^ * $, which provides locally accurate approximation of $p(x_d)$ over the small region ${\Omega _{dj}}$, but as the reference trajectory moves away from ${\Omega _{dj}}$, the weight subvector ${{\hat W}_j}$ will gradually lose the convergence tendency to $W_j^ * $ and deviate from its previous converged value. Therefore, the learning process of the SGD based NNLC is always accompanied by the passive knowledge forgetting phenomenon, which is also the reason why its learning speed is limited. The passive knowledge forgetting also provides an insight into the necessity of high-level PE: high-level PE ensures that forgotten knowledge is persistently reviewed, while low-level PE does not.

\begin{remark}\label{remark4}
Interestingly, the influence of passive knowledge forgetting is related to hyperparameter settings of the RBFNN. When the receptive field widths ${\sigma _i},i = 1,2, \ldots ,N$ are designed to be relatively small, the intersections of the two approximation domains in Fig. \ref{fig1} will become very small such that the learning along $\varphi_{d2}$ hardly affects the learned knowledge along $\varphi_{d1}$. On the contrary, when ${\sigma _i}$ is relatively large, the passive knowledge forgetting phenomenon will become severe such that the convergence of the weight vector becomes very slow and that is the reason why the SGD based NNLC is sensitive to the setting of the receptive field widths. 
\end{remark}

\section{Real-Time Progressive Learning}\label{section5}
In this section, a filter-free learning control scheme named RTPL is proposed based on SMRLS \cite{r25} to improve the performance of the SGD based NNLC mentioned in Section \ref{section4}. Theoretical analysis demonstrates the exponential stability of the closed-loop system under Assumption \ref{assumption3}. 

\subsection{Learning From Feedback Control}\label{section5-1}
Although SMRLS can suppress the passive knowledge forgetting phenomenon in the real-time learning process, it is originally proposed to address the supervised learning problem in open-loop systems where the function value to be approximated can be measured directly \cite{r25}. However, accurate measurement of the function value is usually unachievable in closed-loop systems. 

According to iterative learning control, we can set the feedforward controller output of the current iteration to a linear combination of tracking errors and feedforward output of the last iteration such that the control performance is improved \cite{r25_1, r25_2, r25_3}. Inspired by this idea, the desired output of the RBFNN at time $t$ is set to the sum of measurable tracking errors and the output of the RBFNN itself. Substitute $W = \hat W$, $\chi  = {x_d}$ and $h = u$ into \eqref{eq17}, and the weight update law of RTPL is expressed as  
\begin{equation}\label{eq36}
\left\{ \begin{gathered}
	\dot {\hat W} = P\Phi ({x_d})\left[ {F - {{\hat W}^T}\Phi ({x_d})} \right] - {\epsilon _a} \hfill \\
	{{\dot P}^{ - 1}} = \Phi ({x_d}){\Phi ^T}({x_d}) - {\Phi _a}\Phi _a^T \hfill \\ 
\end{gathered}  \right.
\end{equation}
where $F = \eta {e_n} + {{\hat W}^T}\Phi ({x_d})$ is the estimated desired output of the RBFNN feedforward controller with $\eta>0$ being a designed gain, ${\epsilon _a} = P{\Phi _a}({F_a} - {\hat W^T}{\Phi _a})$ with $F_a$ being the latest record of $F$ within the current partition before the current time. For convenience of theoretical analysis, substitute \eqref{eq26} into \eqref{eq36} and the update law becomes  
\begin{equation}\label{eq37}
	\left\{ \begin{gathered}
		\dot {\hat W} = \eta P\Phi ({x_d}){e_n} - {\epsilon _a} \hfill \\
		{{\dot P}^{ - 1}} = \Phi ({x_d}){\Phi ^T}({x_d}) - {\Phi _a}\Phi _a^T.  \hfill \\ 
	\end{gathered}  \right.
\end{equation}

Fig. \ref{fig2} shows the closed-loop structure of the RTPL based hybrid feedforward feedback control. 
\begin{figure}[htbp]
	\centering
	\includegraphics[scale=0.5]{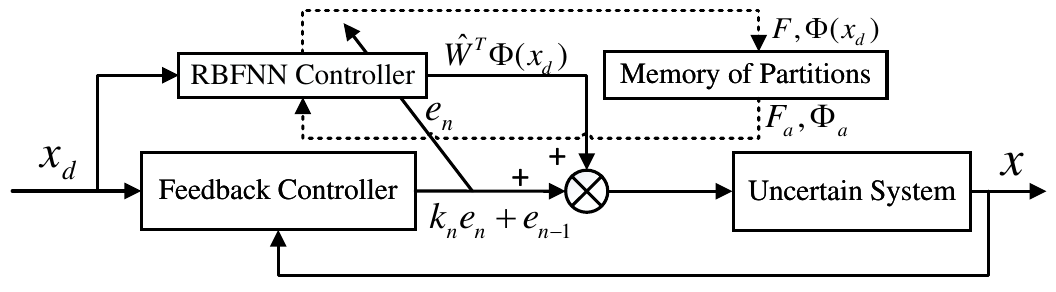}
	\caption{RTPL based hybrid feedforward feedback control}\vspace{-1em}\label{fig2}
\end{figure}

\begin{remark}\label{remark5}
Although SMRLS can effectively improve the learning performance of the RBFNN, its computational complexity is higher than that of the SGD algorithm. In addition, the implementation of RTPL requires extra memory to store the past data. For RTPL based control scheme in Fig. \ref{fig2}, each partition of the input space, i.e. $\Omega^j_d, j=1,2,\ldots,N_P$, needs to store the latest regressor vector $\Phi_j$ and desired output $F_j$ within the partition before the current time \cite{r25}. 
\end{remark}

\subsection{Exponential Stability under Partial PE Condition}\label{section5-2}
With Assumption \ref{assumption3}, only the PE condition of $\Phi_\zeta(x_d)$ is guaranteed, so the update of $\hat{W}_\zeta$ needs to be analyzed separately. Through elementary row transformation, $\hat{W}$ and $\Phi(x_d)$ can be transformed into the following form  
\begin{equation}\label{eq38}
	\left\{ \begin{gathered}
		{S_\zeta }\hat W = \left[ {\begin{array}{*{20}{c}}
				{{{\hat W}_\zeta }} \\ 
				{{{\hat W}_{\bar \zeta }}} 
		\end{array}} \right] \hfill \\
		{S_\zeta }\Phi ({x_d}) = \left[ {\begin{array}{*{20}{c}}
				{{\Phi _\zeta }({x_d})} \\ 
				{{\Phi _{\bar \zeta }}({x_d})} 
		\end{array}} \right] \hfill \\ 
	\end{gathered}  \right.
\end{equation}
where ${S_\zeta } \in {\mathbb{R}^{N \times N}}$ is an orthogonal constant matrix satisfying $S_\zeta ^{ - 1} = S_\zeta ^T$, and the subscript ${\bar \zeta }$ represents the neurons which are far from the trajectory of $x_d$ and barely affect the approximation. Therefore, the following equations are obtained 
\begin{equation}\label{eq39}
	\left\{ \begin{gathered}
		\left[ {\begin{array}{*{20}{c}}
				{{{\dot {\tilde W}}_\zeta }} \\ 
				{{{\dot {\tilde W}}_{\bar \zeta }}} 
		\end{array}} \right] = {S_\zeta }\dot {\tilde W} \hfill \\
		\Phi ({x_d}) = S_\zeta ^{ - 1}\left[ {\begin{array}{*{20}{c}}
				{{\Phi _\zeta }({x_d})} \\ 
				{{\Phi _{\bar \zeta }}({x_d})} 
		\end{array}} \right] \hfill \\ 
	\end{gathered}  \right.
\end{equation}
where ${{\tilde W}_{\bar \zeta }} = W_{\bar \zeta }^ *  - {{\hat W}_{\bar \zeta }}$. Substitute $\tilde{W}=W^* - \hat{W}$ into \eqref{eq37} and we obtain  
\begin{equation}\label{eq40}
\dot{\tilde{W}}=-\dot {\hat W} = -\eta P\Phi ({x_d})e_n + {\epsilon _a} 
\end{equation}

Substitute \eqref{eq40} into \eqref{eq39} and the following weight update law is obtained  
\begin{equation}\label{eq41}
\begin{aligned}
	\left[ {\begin{array}{*{20}{c}}
			{{{\dot {\tilde W}}_\zeta }} \\ 
			{{{\dot {\tilde W}}_{\bar \zeta }}} 
	\end{array}} \right] &=  - \eta {S_\zeta }PS_\zeta ^{ - 1}\left[ {\begin{array}{*{20}{c}}
			{{\Phi _\zeta }({x_d})} \\ 
			{{\Phi _{\bar \zeta }}({x_d})} 
	\end{array}} \right]e_n + {S_\zeta }{\epsilon _a} \hfill \\
	&=  - \eta {S_\zeta }PS_\zeta ^T\left[ {\begin{array}{*{20}{c}}
			{{\Phi _\zeta }({x_d})} \\ 
			{{\Phi _{\bar \zeta }}({x_d})} 
	\end{array}} \right]e_n + {S_\zeta }{\epsilon _a} \hfill \\
	&=  - \eta {P_S}\left[ {\begin{array}{*{20}{c}}
			{{\Phi _\zeta }({x_d})} \\ 
			{{\Phi _{\bar \zeta }}({x_d})} 
	\end{array}} \right]e_n + {S_\zeta }{\epsilon _a} \hfill \\ 
\end{aligned} 
\end{equation}
where ${P_S} = {S_\zeta }PS_\zeta ^T$ is a positive definite matrix.

According to Property \ref{property2}, the values of ${{\Phi _{\bar \zeta }}({x_d})}$ are close to zero during the learning control process. In this case, we have  
\begin{equation}\label{eq42}
	{{\dot {\tilde W}}_\zeta } =  - \eta {P_\zeta }{\Phi _\zeta }({x_d})e_n + {\epsilon _W}
\end{equation}
where ${{\epsilon _W}\in {\mathbb{R}^{{N_\zeta }}}}$ satisfying $\left\| {{\epsilon _W}} \right\| < \epsilon _W^ * $ is a small filtered approximation error, and ${P_\zeta } \in {\mathbb{R}^{{N_\zeta } \times {N_\zeta }}}$ is a principal submatrix of $P_S$. Since $P_S$ is positive definite, ${P_\zeta }$ is also positive definite. 

Replace the weight update law of \eqref{eq29} with \eqref{eq42}, and the closed-loop dynamics of the RTPL based system is obtained  
\begin{equation}\label{eq43}
	\left\{ \begin{aligned}
		&{{\dot e}_1} =  - {k_1}{e_1} + {e_2} \hfill \\
		&{{\dot e}_i} =  - {k_i}{e_i} + {e_{i + 1}} - {e_{i - 1}}, i=2,3,\ldots,n-1 \hfill \\
		&{{\dot e}_n} =  - g(x)\left( {{k_n}{e_n} + {e_{n - 1}} - \tilde W_\zeta ^T{\Phi _\zeta }({x_d}) - {\epsilon _p}(x,{x_d})} \right) \hfill \\
		&{{\dot {\tilde W}}_\zeta } =  - \eta {P_\zeta }{\Phi _\zeta }({x_d})e_n + {\epsilon _W} \hfill \\ 
	\end{aligned}  \right.
\end{equation}

The matrix form of \eqref{eq43} is formulated as  
\begin{equation}\label{eq44}
	\begin{aligned}
		\left[ {\begin{array}{*{20}{c}}
				{\dot e} \\ 
				{{{\dot {\tilde W}}_\zeta }} 
		\end{array}} \right] =& \left[ {\begin{array}{*{20}{c}}
				{A}&\vline & { bg(x)\Phi _\zeta ^T({x_d})} \\ 
				\hline
				{ - \eta {P _\zeta }{\Phi _\zeta }({x_d}){b^{T}}}&\vline & {{0_{{N_\zeta } \times {N_\zeta }}}} 
		\end{array}} \right]\left[ {\begin{array}{*{20}{c}}
				e \\ 
				{{{\tilde W}_\zeta }} 
		\end{array}} \right] \\
		&+ \left[ {\begin{array}{*{20}{c}}
				{ bg(x){\epsilon _p}(x,{x_d})} \\ 
				{\epsilon_W} 
		\end{array}} \right]
	\end{aligned}
\end{equation}
where $A$, $b$ follow the same definition as in \eqref{eq30}. The main difference between \eqref{eq30} and \eqref{eq44} is the time varying gain matrix $P _\zeta$. 

\begin{lemma}\label{lemma3}
Consider the positive definite gain matrix $P_\zeta$ of \eqref{eq42}. There exists a constant $\sigma  > 0$ such that ${X^T}\dot P_\zeta ^{ - 1}X \leqslant \sigma {\left\| X \right\|^2}$ holds for any $X \in {\mathbb{R}^{{N_\zeta }}}$. 
\end{lemma}

\begin{proof}
To prove this lemma, we only need to determine an upper bound of the maximum eigenvalue of $\dot{P}_\zeta^{-1}$. Since ${P_S} = {S_\zeta }PS_\zeta ^T$, where $S_\zeta$ is an orthogonal constant matrix and $P$ is a positive definite matrix, we obtain  
\begin{equation}\label{eq45}
	\dot P_S^{ - 1} = {S_\zeta }{{\dot P}^{ - 1}}S_\zeta ^T. 
\end{equation}
Substitute the update law of $P$ in \eqref{eq37} into \eqref{eq45}  
\begin{equation}\label{eq46}
	\dot P_S^{ - 1} = {S_\zeta }\Phi ({x_d}){\Phi ^T}({x_d})S_\zeta ^T - {S_\zeta }{\Phi _a}\Phi _a^TS_\zeta ^T. 
\end{equation}
As ${P_\zeta } \in {\mathbb{R}^{{N_\zeta } \times {N_\zeta }}}$ is a principal submatrix of $P_S$, it is assumed that ${P_\zeta } = {U_\zeta }{P_S}U_\zeta ^T$ where the constant matrix ${U_\zeta } = \left[ {\begin{array}{*{20}{c}}
		{{I_{{N_\zeta } \times {N_\zeta }}}}&{{0_{{N_\zeta } \times (N - {N_\zeta })}}} 
\end{array}} \right] \in {\mathbb{R}^{{N_\zeta } \times N}}$. To obtain $\dot P_\zeta ^{ - 1}$, consider the following equation  
\begin{equation}\label{eq47}
	{P_\zeta }P_\zeta ^{ - 1} = {I_{{N_\zeta } \times {N_\zeta }}}. 
\end{equation}
Take derivative of both sides of \eqref{eq48} and we have  
\begin{equation}\label{eq48}
	{{\dot P}_\zeta }P_\zeta ^{ - 1} + {P_\zeta }\dot P_\zeta ^{ - 1} = {0_{{N_\zeta } \times {N_\zeta }}}. 
\end{equation}
Therefore, $\dot P_\zeta ^{ - 1}$ can be calculated with  
\begin{equation}\label{eq49}
\begin{aligned}
	\dot P_\zeta ^{ - 1} &=  - P_\zeta ^{ - 1}{{\dot P}_\zeta }P_\zeta ^{ - 1} \hfill \\
	&=  - P_\zeta ^{ - 1}{U_\zeta }{{\dot P}_S}U_\zeta ^TP_\zeta ^{ - 1} \hfill \\ 
	&= P_\zeta ^{ - 1}{U_\zeta }{P_S}\dot P_S^{ - 1}{P_S}U_\zeta ^TP_\zeta ^{ - 1}. \\
\end{aligned} 
\end{equation}

Consider \eqref{eq46} and reformulate $\dot{P}^{-1}_S$ into the following block matrix  
\begin{equation}\label{eq50}
\begin{aligned}
	\dot P_S^{ - 1} &= \left[ {\begin{array}{*{20}{c}}
			{{B_1}}&{{B_2}} \\ 
			{{B_3}}&{{B_4}} 
	\end{array}} \right] \hfill \\
	&= \left[ {\begin{array}{*{20}{c}}
			{{B_1}}&0 \\ 
			0&0 
	\end{array}} \right] + \left[ {\begin{array}{*{20}{c}}
			0&{{B_2}} \\ 
			{{B_3}}&{{B_4}} 
	\end{array}} \right] \hfill \\
	&= \dot P_{S\zeta }^{ - 1} + \dot P_{S\bar \zeta }^{ - 1} \hfill \\ 
\end{aligned} 
\end{equation}
where ${B_1} = {\Phi _\zeta }({x_d})\Phi _\zeta ^T({x_d}) - {\Phi _{a\zeta }}\Phi _{a\zeta }^T \in {\mathbb{R}^{{N_\zeta } \times {N_\zeta }}}$ with ${\Phi _{a\zeta }}$ being the recorded regressor subvector corresponding to ${\Phi _\zeta }({x_d})$, $\dot P_{S\zeta }^{ - 1} = \left[ {\begin{array}{*{20}{c}}
		{{B_1}}&0 \\ 
		0&0 
\end{array}} \right]$, $\dot P_{S\bar \zeta }^{ - 1} = \left[ {\begin{array}{*{20}{c}}
		0&{{B_2}} \\ 
		{{B_3}}&{{B_4}} 
\end{array}} \right]$, and the entries of $B_2$, $B_3$, $B_4$ are all close to zero. 

Substitute \eqref{eq50} into \eqref{eq49} and we obtain 
\begin{equation}\label{eq51}
\begin{aligned}
	\dot P_\zeta ^{ - 1} &= P_\zeta ^{ - 1}{U_\zeta }{P_S}(\dot P_{S\zeta }^{ - 1} + \dot P_{S\bar \zeta }^{ - 1}){P_S}U_\zeta ^TP_\zeta ^{ - 1} \hfill \\
	&= {B_1} + P_\zeta ^{ - 1}{U_\zeta }{P_S}\dot P_{S\bar \zeta }^{ - 1}{P_S}U_\zeta ^TP_\zeta ^{ - 1} \hfill \\ 
\end{aligned} 
\end{equation}
According to the definition of $B_1$, we obtain  
\begin{equation}\label{eq52}
	\begin{small}
\begin{aligned}
	{\lambda _{\max }}(B_1) &= {\lambda _{\max }}\left( {\Phi _\zeta }({x_d})\Phi _\zeta ^T({x_d}) - {\Phi _{a\zeta }}\Phi _{a\zeta }^T \right) \hfill \\
	&= \mathop {\max }\limits_{X \in {\mathbb{R}^N},\left\| X \right\| \ne 0} \frac{{{X^T}\left( {\Phi _\zeta }({x_d})\Phi _\zeta ^T({x_d}) - {\Phi _{a\zeta }}\Phi _{a\zeta }^T \right)X}}{{{X^T}X}} \hfill \\
	&= {\lambda _{\max }}\left( {\Phi_\zeta ({x_d}){\Phi_\zeta ^T}({x_d})} \right) - {\lambda _{\max }}\left( {{\Phi _{a\zeta}}\Phi _{a\zeta}^T} \right) \hfill \\
	&= tr\left( {\Phi_\zeta ({x_d}){\Phi_\zeta ^T}({x_d})} \right) - tr\left( {{\Phi _{a\zeta}}\Phi _{a\zeta}^T} \right) \hfill \\
	&\leqslant N. \hfill \\ 
\end{aligned} 
\end{small}
\end{equation}
Since each entry of $\dot P_{S\bar \zeta }^{ - 1}$ is close to zero, it is shown that each eigenvalue of it is also close to zero according to Gershgorin circle theorem. Therefore, we can assume that there exists a constant $\gamma  > 0$ such that  
\begin{equation}\label{eq53}
{\lambda _{\max }}(\dot P_{S\bar \zeta }^{ - 1}) \leqslant \gamma
\end{equation}
always holds along the trajectory of \eqref{eq46}.

With Assumption \ref{assumption5}, we have  
\begin{equation}\label{eq54}
{\lambda _{\max }}({P_S}) = {\lambda _{\max }}(P) \leqslant {p_0}. 
\end{equation}
Since $P_\zeta$ is a principal submatrix of $P_S$, we obtain ${\lambda _{\min }}({P_\zeta }) \geqslant {\lambda _{\min }}({P_S}) \geqslant {q_0}$ according to Cauchy interlace theorem. Therefore, we have  
\begin{equation}\label{eq55}
\begin{aligned}
	{\lambda _{\max }}(P_\zeta ^{ - 1}) = \frac{1}{{{\lambda _{\min }}({P_\zeta })}} \leqslant \frac{1}{{{q_0}}}. \hfill \\ 
\end{aligned} 
\end{equation}
According to the definition of $U_\zeta$, the following inequation  
\begin{equation}\label{eq56}
\left\| {{X^T}P_\zeta ^{ - 1}{U_\zeta }} \right\| \leqslant \left\| {{X^T}P_\zeta ^{ - 1}} \right\|
\end{equation}
holds for any $X \in {\mathbb{R}^{{N_\zeta }}}$. 

Consider \eqref{eq51}, \eqref{eq52}, \eqref{eq53}, \eqref{eq54}, \eqref{eq55}, \eqref{eq56} and we obtain  
\begin{equation}\label{eq57}
	\begin{small}
\begin{aligned}
	{X^T}\dot P_\zeta ^{ - 1}X &= {X^T}P_\zeta ^{ - 1}{U_\zeta }{P_S}\dot P_S^{ - 1}{P_S}U_\zeta ^TP_\zeta ^{ - 1}X \hfill \\
	&= {X^T}{B_1}X + {X^T}P_\zeta ^{ - 1}{U_\zeta }{P_S}\dot P_{S\bar \zeta }^{ - 1}{P_S}U_\zeta ^TP_\zeta ^{ - 1}X \hfill \\
	&\leqslant {\lambda _{\max }}({B_1}){\left\| X \right\|^2} + {\lambda _{\max }}(\dot P_{S\bar \zeta }^{ - 1}){\left\| {{X^T}P_\zeta ^{ - 1}{U_\zeta }{P_S}} \right\|^2} \hfill \\
	&\leqslant N{\left\| X \right\|^2} + \gamma p_0^2{\left\| {{X^T}P_\zeta ^{ - 1}{U_\zeta }} \right\|^2} \hfill \\
	&\leqslant N{\left\| X \right\|^2} + \gamma p_0^2{\left\| {{X^T}P_\zeta ^{ - 1}} \right\|^2} \hfill \\
	&\leqslant N{\left\| X \right\|^2} + \gamma {\left( {\frac{{{p_0}}}{{{q_0}}}} \right)^2}{\left\| X \right\|^2} \hfill \\ 
\end{aligned} 
\end{small}
\end{equation}
holds for any $X \in {\mathbb{R}^{{N_\zeta }}}$. Let $\sigma  = N + \gamma {\left( {\frac{{{p_0}}}{{{q_0}}}} \right)^2}$ and this ends the proof. 
\end{proof}

\begin{lemma}\label{lemma4}
	Consider the nominal system of \eqref{eq44} as follows  
	\begin{equation}\label{eq58}
		\begin{aligned}
			\left[ {\begin{array}{*{20}{c}}
					{\dot e} \\ 
					{{{\dot {\tilde W}}_\zeta }} 
			\end{array}} \right] = \left[ {\begin{array}{*{20}{c}}
					{A}&\vline & { bg(x)\Phi _\zeta ^T({x_d})} \\ 
					\hline
					{ - \eta { P_\zeta }{\Phi _\zeta }({x_d}){b^T}}&\vline & {{0_{{N_\zeta } \times {N_\zeta }}}} 
			\end{array}} \right]\left[ {\begin{array}{*{20}{c}}
					e \\ 
					{{{\tilde W}_\zeta }} 
			\end{array}} \right] 
		\end{aligned}. 
	\end{equation}
	With Assumption \ref{assumption1}, the PE condition of ${{\Phi _\zeta }({x_d})}$ and sufficiently large control gains $\eta, {k_i},i = 1,2, \ldots ,n$, \eqref{eq58} is exponentially stable, which means that the errors $e$ and ${\tilde W}_\zeta$ will exponentially converge to zero. 
\end{lemma}

\begin{proof}
Apply the linear transformation  
\begin{equation}\label{eq59}
	\left\{ \begin{gathered}
		{z_i} = {e_i},i = 1,2, \ldots ,n - 1 \hfill \\
		{z_n} = {e_n}/g(x) \hfill \\ 
	\end{gathered}  \right.
\end{equation}
to \eqref{eq55} and we obtain  
\begin{equation}\label{eq60}
\left[ {\begin{array}{*{20}{c}}
		{\dot z} \\ 
		{{{\dot {\tilde W}}_\zeta }} 
\end{array}} \right] = \left[ {\begin{array}{*{20}{c}}
		B&\vline & {b\Phi _\zeta ^T({x_d})} \\ 
		\hline
		{ - \eta g(x){P_\zeta }{\Phi _\zeta }({x_d}){b^T}}&\vline & {{0_{{N_\zeta } \times {N_\zeta }}}} 
\end{array}} \right]\left[ {\begin{array}{*{20}{c}}
		z \\ 
		{{{\tilde W}_\zeta }} 
\end{array}} \right]
\end{equation}
where 
\begin{equation}\label{eq61}
B = \left[ {\begin{array}{*{20}{c}}
		{ - {k_1}}&1&{}&{}&{} \\ 
		{ - 1}&{ - {k_2}}&1&{}&{} \\ 
		{}&{}& \ddots &{}&{} \\ 
		{}&{}&{ - 1}&{ - {k_{n - 1}}}&{g(x)} \\ 
		{}&{}&{}&{ - 1}&{ - {k_n}g(x) - \frac{{\dot g(x)}}{{g(x)}}} 
\end{array}} \right]. 
\end{equation}

Let $H = \left[ {\begin{array}{*{20}{c}}
		{{I_{(n - 1) \times (n - 1)}}}&{} \\ 
		{}&{g(x)} 
\end{array}} \right]$, ${B^T}H + {H^T}B + \dot H =  - Q$, and we have  
\begin{equation}\label{eq62}
Q = \left[ {\begin{array}{*{20}{c}}
		{2{k_1}}&{}&{}&{}&{} \\ 
		{}&{2{k_2}}&{}&{}&{} \\ 
		{}&{}& \ddots &{}&{} \\ 
		{}&{}&{}&{2{k_{n - 1}}}&{} \\ 
		{}&{}&{}&{}&{ 2{k_n}{g^2}(x) + \dot g(x)}
\end{array}} \right]. 
\end{equation}
Thus we can guarantee that $Q$ is positive definite by choosing  
\begin{equation}\label{eq63}
\left\{ \begin{gathered}
	{k_i} > 0,i = 1,2, \ldots ,n - 1 \hfill \\
	{k_n} >  \frac{{{g_{1u}}}}{{2g_{0l}^2}}. \hfill \\ 
\end{gathered} \right.
\end{equation}

Consider the following Lyapunov function candidate  
\begin{equation}\label{eq64}
{V_1} = \frac{1}{{2}}{z^T}Hz + \frac{1}{{2\eta }}{{\tilde W}^T}P_\zeta ^{ - 1}\tilde W. 
\end{equation}
The trajectory of $\dot{V}_1$ along \eqref{eq58} satisfies  
\begin{equation}\label{eq65}
\begin{aligned}
{{\dot V}_1} &=  - \frac{1}{2}{z^T}Qz + \frac{1}{{2\eta }}{{\tilde W}^T}\dot P_\zeta ^{ - 1}\tilde W \\
&\leqslant  - \frac{1}{2}{\lambda _{\min }}(Q){\left\| z \right\|^2} + \frac{\sigma }{{2\eta }}{\left\| {\tilde W} \right\|^2}
\end{aligned}
\end{equation}
where $\sigma>0$ is an estimate of the maximum eigenvalue of $\dot P_\zeta ^{ - 1}$ shown in Lemma \ref{lemma3}. 

Let $\Theta  = \left[ {\begin{array}{*{20}{c}}
		z \\ 
		{{{\tilde W}_\zeta }} 
\end{array}} \right]$ and the system \eqref{eq60} is decomposed into the following form 
\begin{equation}\label{eq66}
\dot \Theta  = \Xi \Theta  + \Psi 
\end{equation}
where 
\begin{equation}\label{eq67}
	\begin{small}
	\begin{aligned}
&\Xi  = \left[ {\begin{array}{*{20}{c}}
		B&\vline & {b\Phi _\zeta ^T({x_d})} \\ 
		\hline
		{ - \eta g(x){P_\zeta }(0){\Phi _\zeta }({x_d}){b^T}}&\vline & {{0_{{N_\zeta } \times {N_\zeta }}}} 
\end{array}} \right], \\
&\Psi  = \left[ {\begin{array}{*{20}{c}}
		{{0_n}} \\ 
		{\eta g(x)\left( {{P_\zeta }(0) - {P_\zeta }} \right){\Phi _\zeta }({x_d}){z_n}} 
\end{array}} \right]. 
\end{aligned}
\end{small}
\end{equation}
According to Lemma 1 in \cite{r35}, the nominal system $\dot \Theta  = \Xi \Theta $ is exponentially stable with the PE condition of ${{\Phi _\zeta }({x_d})}$. As analyzed in Theorem 4.12 of \cite{r36}, there exists a positive definite matrix $\Lambda$ such that ${\Lambda ^T}\Xi  + {\Xi ^T}\Lambda  + \dot \Lambda  =  - D$ where $D$ is a given positive definite matrix. Consider the following Lyapunov function candidate  
\begin{equation}\label{eq68}
{V_2} = \frac{1}{2}{\Theta ^T}\Lambda \Theta 
\end{equation}
Thus the derivative of $V_2$ along the trajectory of \eqref{eq66} satisfies  
\begin{equation}\label{eq69}
	\begin{small}
\begin{aligned}
	{{\dot V}_2} &= \frac{1}{2}\left( {{\Theta ^T}\Lambda \dot \Theta  + {{\dot \Theta }^T}\Lambda \Theta  + {\Theta ^T}\dot \Lambda \Theta } \right) \hfill \\
	&=  - \frac{1}{2}{\Theta ^T}D\Theta  + {\Theta ^T}\Lambda \Psi  \hfill \\
	&\leqslant  - \frac{1}{2}{\lambda _{\min }}(D){\left\| \Theta  \right\|^2} + {\lambda _{\max }}(\Lambda )\left\| \Theta  \right\|\left\| \Psi  \right\| \hfill \\ 
	& \leqslant  - \frac{1}{2}{\lambda _{\min }}(D){\left\| \Theta  \right\|^2} + {\lambda _0}\left\| \Theta  \right\|\left\| z \right\|
\end{aligned} 
\end{small}
\end{equation}
where $\lambda_0$ is an estimate of the upper bound of ${\lambda _{\max }}(\Lambda )\left\| {\eta g(x)\left( {{P_\zeta }(0) - {P_\zeta }} \right){\Phi _\zeta }({x_d})} \right\|$. Since $P_\zeta $ satisfies ${\lambda _{\min }}({P_\zeta }) \geqslant {q_0}$, ${\lambda _0} = \eta {g_{0u}}({p_0} - {q_0})N{\lambda _{\max }}(\Lambda )$ is an appropriate choice. 

Consider the following Lyapunov function candidate composed of $V_1$, $V_2$  
\begin{equation}\label{eq70}
{V_3} = \pi {V_1} + {V_2}. 
\end{equation}
The derivative of $V_3$ along \eqref{eq60} satisfies  
\begin{equation}\label{eq71}
	\begin{footnotesize}
	\begin{aligned}
{{\dot V}_3} \leqslant  - \frac{\pi }{2}{\lambda _{\min }}(Q){\left\| z \right\|^2} + \frac{{\pi \sigma }}{{2\eta }}{\left\| {\tilde W} \right\|^2} - \frac{1}{2}{\lambda _{\min }}(D){\left\| \Theta  \right\|^2} + {\lambda _0}\left\| \Theta  \right\|\left\| z \right\|. 
\end{aligned}
\end{footnotesize}	
\end{equation}
By choosing $\pi  = \frac{{2\lambda _0^2}}{{{\lambda _{\min }}(Q){\lambda _{\min }}(D)}}$, we obtain  
\begin{equation}\label{eq72}
{{\dot V}_3} \leqslant  - \frac{1}{4}{\lambda _{\min }}(D){\left\| \Theta  \right\|^2} + \frac{{\pi \sigma }}{{2\eta }}{\left\| {\tilde W} \right\|^2}
\end{equation}
Consider \eqref{eq62}, \eqref{eq72}, and we can increase the gains $\eta, {k_i},i = 1,2, \ldots ,n$ such that  
\begin{equation}\label{eq73}
 - \frac{1}{8}{\lambda _{\min }}(D) + \frac{{\pi \sigma }}{{2\eta }} \leqslant 0. 
\end{equation}
Substitute \eqref{eq73} into \eqref{eq72} and we obtain  
\begin{equation}\label{eq74}
{{\dot V}_3} \leqslant  - \frac{1}{8}{\lambda _{\min }}(D){\left\| \Theta  \right\|^2}. 
\end{equation}

Therefore, when the gains $\eta, {k_i},i = 1,2, \ldots ,n$ are designed to be sufficiently large such that \eqref{eq63} and \eqref{eq73} are satisfied, the system \eqref{eq58} is exponentially stable. 
\end{proof}

\begin{theorem}\label{theorem2}
Consider the system \eqref{eq44}. With Assumption \ref{assumption1}, the PE condition of ${{\Phi _\zeta }({x_d})}$ and properly designed control gains, $e$ and ${\tilde W}_\zeta$ will exponentially converge to small neighborhoods around zero, which can be arbitrarily contracted with the increase of $\eta, {k_i},i = 1,2, \ldots ,n$. 
\end{theorem}

\begin{proof}
Similar to the proof of Theorem 1, the perturbed term $\left[ {\begin{array}{*{20}{c}}
		{ bg(x){\epsilon _p}(x,{x_d})} \\ 
		{\epsilon_W} 
\end{array}} \right]$ is decomposed into a nonvanishing term $\left[ {\begin{array}{*{20}{c}}
{bg(x)({\epsilon _{p1}}({x_d}) + {\epsilon _\zeta }({x_d}))} \\ 
{{\epsilon _W}} 
\end{array}} \right]$ and a vanishing term $\left[ {\begin{array}{*{20}{c}}
{bg(x){\epsilon _{p2}}(x,{x_d})} \\ 
{{0_{{N_\zeta } \times 1}}} 
\end{array}} \right]$. 

According to Lemma \ref{lemma4}, the nominal part of \eqref{eq44} is exponentially stable with sufficiently large gains $\eta, {k_i},i = 1,2, \ldots ,n$. Thus Lemma 9.1 and 9.2 in \cite{r36} indicate that the system \eqref{eq44} is semiglobally exponentially stable with properly designed control gains. Specifically, $e$ and ${\tilde W}_\zeta$ will exponentially converge to small neighborhoods around zero, which can be arbitrarily contracted with the increase of $\eta, {k_i},i = 1,2, \ldots ,n$. 
\end{proof}

Assume that the RTPL based learning control happens on the time interval $\left[ {0,T} \right]$. Since the passive knowledge forgetting phenomenon is suppressed, the RBFNN will not pay excessive attention to the latest data, and the integral knowledge extraction method \eqref{eq33} is abandoned. The weight vector $\hat{W}(T)$ is directly regarded as the learned knowledge. In this case, the unknown function $p(x_d)$ is approximated by  
\begin{equation}\label{eq75}
\hat {p}({x_d}) = \hat{W}^T(T)\Phi ({x_d}) 
\end{equation}
which can be used to design the feedforward controller in subsequent control tasks. 

\begin{remark}\label{remark6}
Consider the RTPL based control law \eqref{eq26} and weight update law \eqref{eq42}, an equivalent adaptive gain is defined as follows  
\begin{equation}\label{eq76}
{K_E}({x_d}) \triangleq   \eta \Phi _\zeta ^T({x_d}){P_\zeta }{\Phi _\zeta }({x_d})
\end{equation}
which reflects the sensitivity of the output of the RBFNN to the error signal $e_n$. With the initialization in Remark \ref{remark3}, ${K_E}({x_d})$ is rather large at the beginning of the learning control process. To avoid the drastic update at the beginning, a saturation-like gain $\eta(t)$ is proposed  
\begin{equation}\label{eq77}
\eta(t) = \left\{ \begin{aligned}
	&\frac{\eta_0 t}{{{T_0}}},&\text{ }&for \text{ }0 \leqslant t \leqslant {T_0} \hfill \\
	&\text{ }\eta_0,&\text{ }&for \text{ }t > {T_0} \hfill \\ 
\end{aligned}  \right.
\end{equation}
where $T_0>0$ is a constant. At the transient phase of the learning control process, the desired output of the neural network estimated with tracking error is not accurate enough. Therefore, another reason for employing such a time-varying gain is to avoid recording data which has significant deviation from the desired output into the memory of partitions. It should be noted that there are many alternative designs of $\eta(t)$ besides \eqref{eq72} that may achieve good performance. 
\end{remark}

\subsection{Comparison with the SGD Based Learning Control}\label{section5-3}
The objective of RTPL is to improve control performance and learn unknown dynamics of the system simultaneously. In this section, the merits and defects of RTPL are analyzed in detail through comparison with the SGD based RBFNN learning control. 

\noindent \textbf{Merits of RTPL: }

1) {The learning speed is improved without using filters such that the learning can be achieved in both repetitive and nonrepetitive control tasks. Having overcome passive knowledge forgetting, the neural network will not forget the learned knowledge even in a nonrepetitive control task. }

2) {The robustness to hyperparameter setting of the neural network is improved. The SGD based NNLC is more sensitive to hyperparameter setting of the RBFNN because large learning rate and receptive field width will aggravate passive knowledge forgetting while small ones will degrade the learning performance \cite{r25}. }

3) {The generalization ability of the learned knowledge is improved as a result of the memory mechanism which allocates the synthesized samples approximately uniformly to the normalized feature space \cite{r25}. }

4) {The guaranteed learning performance under parameter perturbation. With SMRLS, the data that does not match the current dynamics will be efficiently replaced by the data representing the the current dynamics. }

\begin{remark}\label{remark7}
The most fundamental advantage of RTPL compared to other learning control algorithms lies in its capability to progressively accumulate knowledge of unknown dynamics during the real-time learning process. This characteristic implies that as the reference trajectory $x_d$ traverses more regions in the input space of the neural network, the learned knowledge becomes applicable to more control tasks. 
\end{remark}

\noindent \textbf{Defects of RTPL: }
The main defect of RTPL is the increased computational complexity compared with the SGD based learning control. As analyzed in \cite{r25}, dimensionality reduction methods such as principal component analysis can be used to reduce the computational complexity when the input dimension $n$ is high. 

\section{Simulation Studies}\label{section6}
In this section, simulation studies are carried out on the same single inverted pendulum-cart system as in \cite{r25} to demonstrate the effectiveness of RTPL. The state space description of the system is formulated as follows  
\begin{equation}\label{eq78}
\left\{ \begin{aligned}
	{{\dot x}_1} &= {x_2} \hfill \\
	{{\dot x}_2} &= {f_I}(x) + {g_I}(x)u \hfill \\
	y \text{ }&= {x} \hfill \\ 
\end{aligned}  \right.
\end{equation}
where 
\begin{equation}\label{eq79}
{f_I}(x) = \frac{{g\sin {x_1} - mlx_2^2\cos {x_1}\sin {x_1}/({m_c} + m)}}{{l(4/3 - m{{\cos }^2}{x_1}/({m_c} + m))}}, 
\end{equation}
\begin{equation}\label{eq80}
{g_I}(x) = \frac{{\cos {x_1}/({m_c} + m)}}{{l(4/3 - m{{\cos }^2}{x_1}/({m_c} + m))}} 
\end{equation}
and the state variables $x_1$ and $x_2$ represent the angle and angular velocity of the pendulum, respectively. The other paramters of the system are set as: the weight of the cart $m_c=0.1kg$, the weight of the pendulum $m=0.02kg$, half the length of the pendulum $l=0.2m$, the gravitational acceleration $g=9.8m/s^2$. In addition, the initial value of $x$ is set to $x_0=\left[\pi/60, 0\right]^T$. 

Since RTPL is a filter-free NNLC method, the filter-based methods will not be considered in this part. The tracking and learning performance of the stochastic gradient descent learning (SGDL) control and RTPL control in the form of hybrid feedforward feedback control are compared. In all of the simulation studies, the control law \eqref{eq26} is adopted whose gains are set to $k_1=2$, $k_2=5$, and the weight update laws \eqref{eq27}, \eqref{eq37} are adopted with the sampling period $\Delta T = 0.005s$. To demonstrate the effect of the RBFNN feedforward controller, a PD feedback controller is also considered by removing the network term from \eqref{eq26}. 

The RBFNNs employed in SGDL and RTPL share the same lattice distribution of neurons as shown in Fig. \ref{fig3}, where there are $5 \times 5$ neurons evenly distributed over the input space $\left[ { - 1,1} \right] \times \left[ { - 1,1} \right]$. For RTPL, the input space of the RBFNN is evenly discritized into $100 \times 100$ partitions based on SMRLS. In addition to aforementioned settings, there are three groups of parameters shown in TABLE \ref{table1} used for contrastive analysis. 

\begin{figure}[htbp]
	\centering
	\includegraphics[scale=0.36]{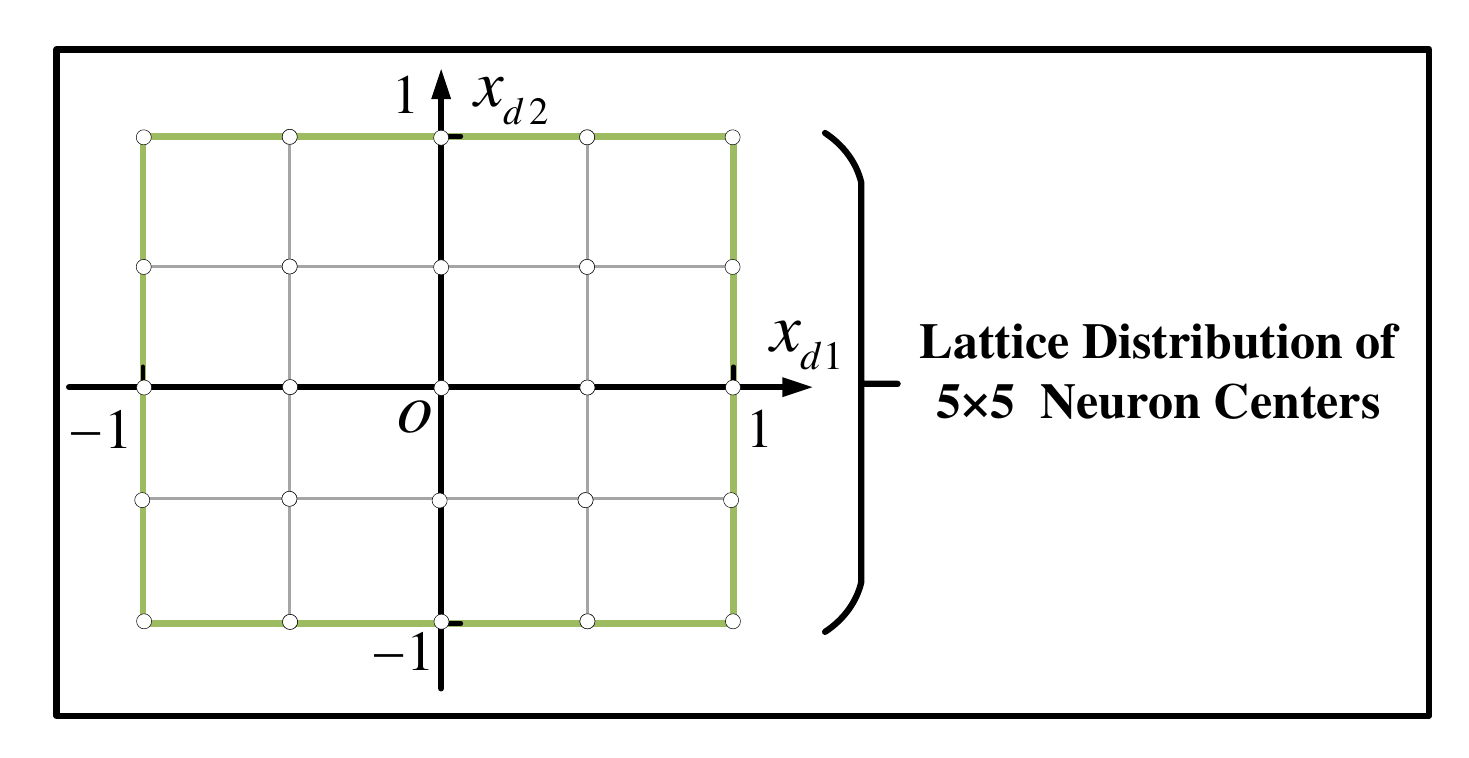}
	\caption{Lattice distribution of the RBFNN neuron centers}\vspace{-1em}\label{fig3}
\end{figure}

\begin{table}
	\centering
	\begin{tabular}{|c|c|c|c|c|}
		\hline
		\multicolumn{2}{|c|}{}&$a)$&$b)$&$c)$ \\
		\hline
		\multirow{4}*{RTPL}&$\sigma_i$&$0.3$&$2$&$0.5$\\
		\cline{2-5}
		&$\eta_0$&$5$&$5$&$5$ \\
		\cline{2-5}
		&$T_0$&$2$&$2$&$2$ \\
		\cline{2-5}
		&$p_0$&$100$&$100$&$100$ \\
		\hline
		\multirow{2}*{SGDL}&$\sigma_i$&$0.3$&$2$&$0.5$\\
		\cline{2-5}
		&$\Gamma$&\text{ }\text{ }$0.1I$\text{ }\text{ }&$0.005I$&\text{ }$0.05I$\text{ }\\
		\hline
	\end{tabular}
		\captionsetup{font={small}}
	\caption{Hyperparameter Settings in Different Tasks}
		\label{table1}
\end{table}

Section \ref{section6-1} and Section \ref{section6-2} compare RTPL and SGDL with different receptive field widths in a repetitive and a nonrepetitive task, respectively. Section \ref{section6-3} demonstrates the performance of RTPL in a repetitive control task in which the parameters of the system are changed abruptly in the learning control process. In Section \ref{section6-4}, the generalization ability of the learned knowledge through RTPL is verified by reusing the knowledge in a different task. Furthermore, the relationship between the training duration on random trajectories and the performance of the acquired weights on test trajectories is analyzed to demonstrate the ability of RTPL in progressively accumulating knowledge. 

\subsection{Learning Control in Repetitive Tasks}\label{section6-1}
Consider the following sinusoidal reference trajectory $\varphi_A$  
\begin{equation}\label{eq81}
\left\{ \begin{gathered}
	{{x}_{d1}} = \sin t \hfill \\
	{{x}_{d2}} = \cos t.  \hfill \\ 
\end{gathered}  \right.
\end{equation}
To demonstrate that the passive knowledge forgetting phenomenon is influenced by the receptive field width $\sigma_i$ as mentioned in Remark \ref{remark4}, the different hyperparameter settings $a)$, $b)$, $c)$ in TABLE \ref{table1} are adopted for SGDL and RTPL. 

Fig. \ref{fig4}, Fig. \ref{fig5} and Fig. \ref{fig6} illustrate the the weight convergence and control performance in the repetitive tracking control task with different hyperparameters a), b), c), respectively. When the receptive field width $\sigma_i$ is small enough and suitbale for SGDL, e.g. $\sigma_i=0.3$ in Fig. \ref{fig4}, both RTPL and SGDL can achieve the weight convergence but the convergence speed of RTPL is obviously higher than SGDL indicating that passive knowledge forgetting still exists in this case. According to Remark \ref{remark4}, the influence of passive knowledge forgetting will increase with the receptive field width. Therefore, in Fig. \ref{fig5} and Fig. \ref{fig6}, the weights of SGDL hardly converge while RTPL still achieves fast weight convergence. Although SGDL also improves the tracking accuracy of the original PD controller, the reason of its effectiveness lies in the high adaptive gain rather than accurate weight convergence. However, the high adaptive gain will decrease the robustness of the system and that is the reason why the weight convergence is emphasized. 

\begin{figure*}[htbp]
	\makebox[\textwidth][c]{\includegraphics[height=0.2\textwidth,width=1.05\textwidth]{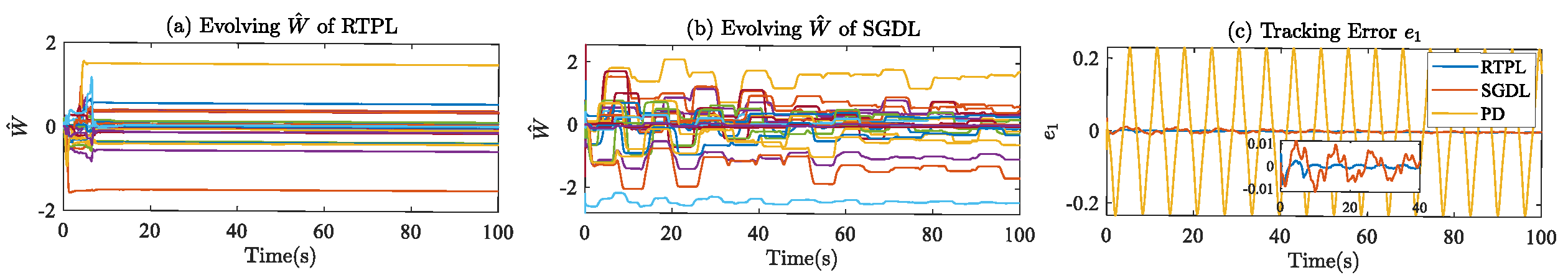}}
	\caption{Learning phase performance along $\varphi_A$ with $\sigma_i=0.3$ }
	\label{fig4}\vspace{-1em}
\end{figure*}

\begin{figure*}[htbp]
	\makebox[\textwidth][c]{\includegraphics[height=0.2\textwidth,width=1.06\textwidth]{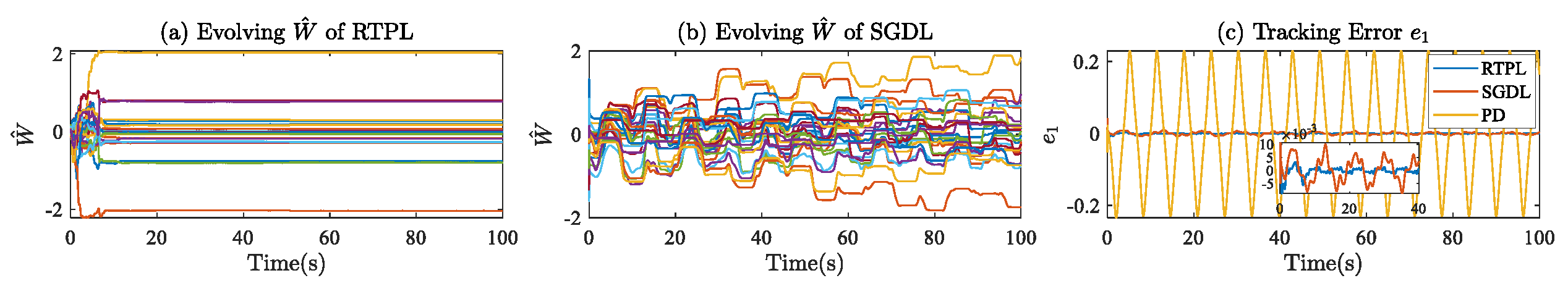}}
	\caption{Learning phase performance along $\varphi_A$ with $\sigma_i=0.5$}
	\label{fig5}\vspace{-1em}
\end{figure*}

\begin{figure*}[htbp]
	\makebox[\textwidth][c]{\includegraphics[height=0.2\textwidth,width=1.06\textwidth]{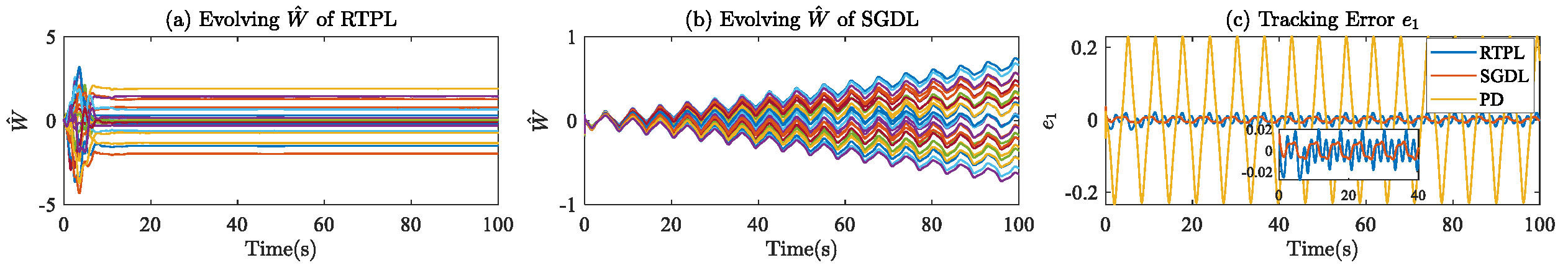}}
	\caption{Learning phase performance along $\varphi_A$ with $\sigma_i=2$}
	\label{fig6}\vspace{-1em}
\end{figure*}

After the learning phase, the learned knowledge of RTPL and SGDL is recorded using \eqref{eq75} and \eqref{eq34} respectively, where the integral interval of \eqref{eq33} is set to the last $5s$ of the learning process. The learned knowledge is reused to the same control task and the performance of the knowledge is shown in Fig. \ref{fig7}. Obviously, the learning speed of RTPL is much higher than SGDL while SGDL with relatively large receptive field width cannot even achieve basic approximation of the unknown function $p(x_d)$ within $100s$. 

\begin{figure*}[htbp]
	\makebox[\textwidth][c]{\includegraphics[height=0.3\textwidth,width=1.03\textwidth]{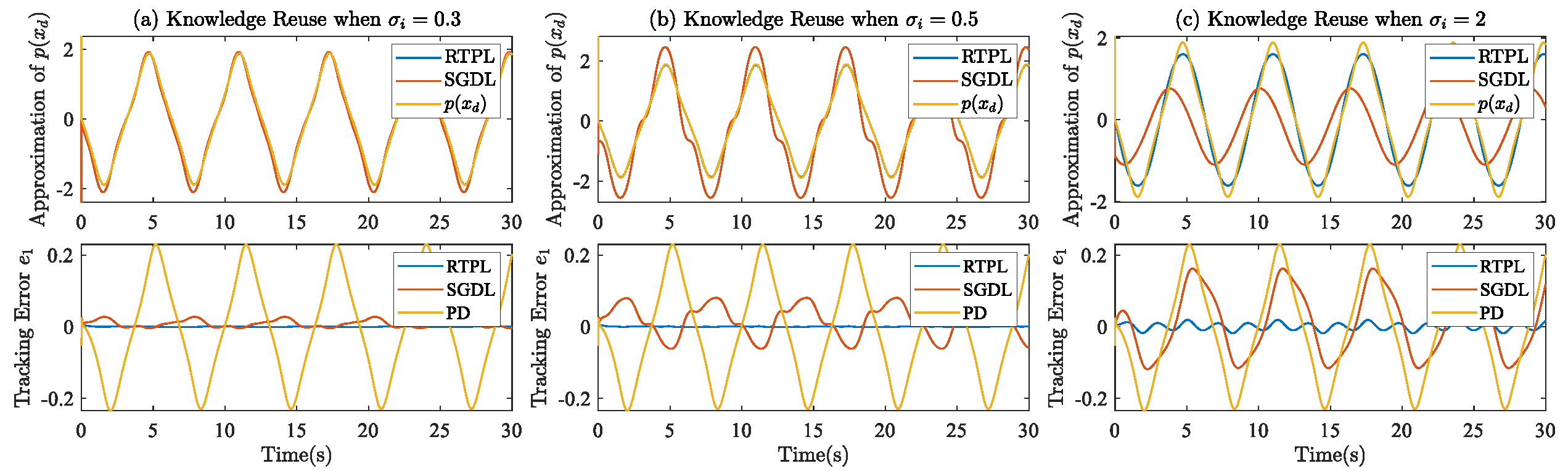}}
	\caption{Approximation and tracking performance of the learned knowledge along $\varphi_A$ with different $\sigma_i$}
	\label{fig7}\vspace{-1em}
\end{figure*}

The results indicate that RTPL achieves higher learning speed compared with SGDL in repetitive tasks. In addition, it is shown that the passive knowledge forgetting phenomenon of SGDL becomes severe with the increase of the receptive field width of the RBFNN. 

\subsection{Learning Control in Nonrepetitive Tasks}\label{section6-2}
To compare the performance of RTPL and SGDL in nonrepetitive tasks, a random non-uniform rational B-splines (NURBS) trajectory $\varphi_B$ with ${x_{d1}} \in \left[ { - 1,1} \right]$ is selected as the reference trajectory. The angular velocity $x_{d2}$ is also normalized into $\left[ { - 1,1} \right]$ for the partitioning operation of RTPL. Fig. \ref{fig8} illustrates the reference trajectory $\varphi_B$. The hyperparameter setting $a)$ in TABLE \ref{table1} is adopted in this task. 
\begin{figure}[htbp]
	\centering
	\includegraphics[scale=0.6]{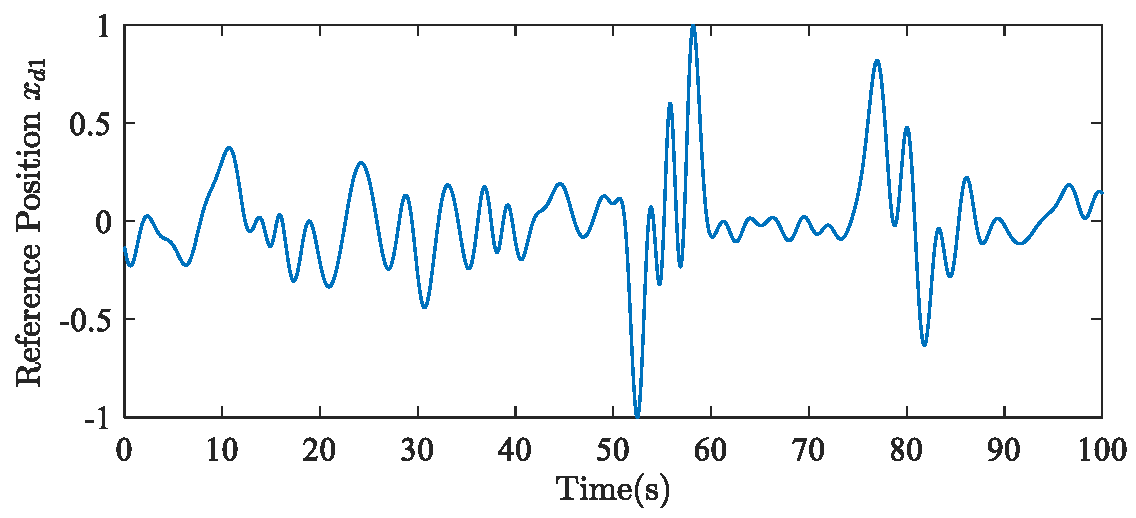}
	\caption{Random NURBS trajectory $\varphi_B$}\vspace{-1em}\label{fig8}
\end{figure}

Fig. \ref{fig9} and Fig. \ref{fig10} illustrate the learning phase and reuse phase performance of RTPL and SGDL, respectively, with the same knowledge extraction mentioned in Section \ref{section6-1}. The results demonstrate that RTPL achieves high learning speed while the SGDL method has not learned anything about the system. These results indicate that the learning performance of SGDL degrades significantly in nonrepetitive tasks as a result of the passive knowledge forgetting phenomenon while RTPL still has good learning performance. 

\begin{figure*}[htbp]
	\makebox[\textwidth][c]{\includegraphics[height=0.2\textwidth,width=1.05\textwidth]{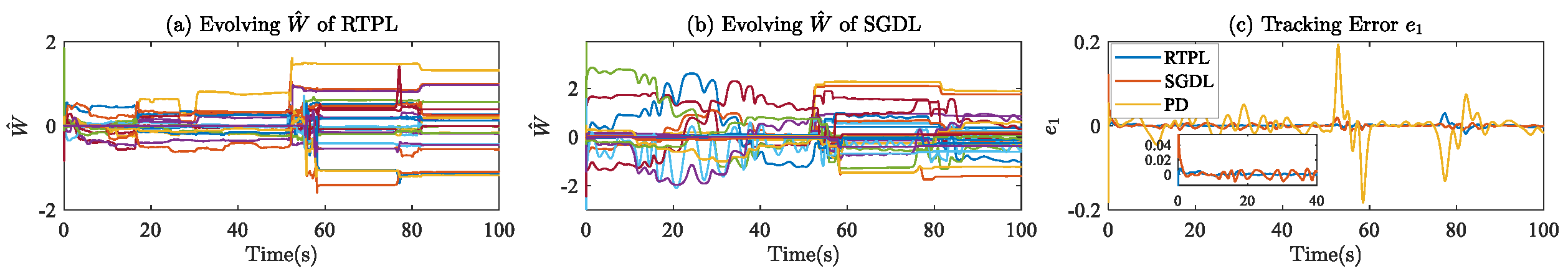}}
	\caption{Learning phase performance along $\varphi_B$}
	\label{fig9}\vspace{-1em}
\end{figure*}

\begin{figure*}[htbp]
	\makebox[\textwidth][c]{\includegraphics[height=0.2\textwidth,width=1.03\textwidth]{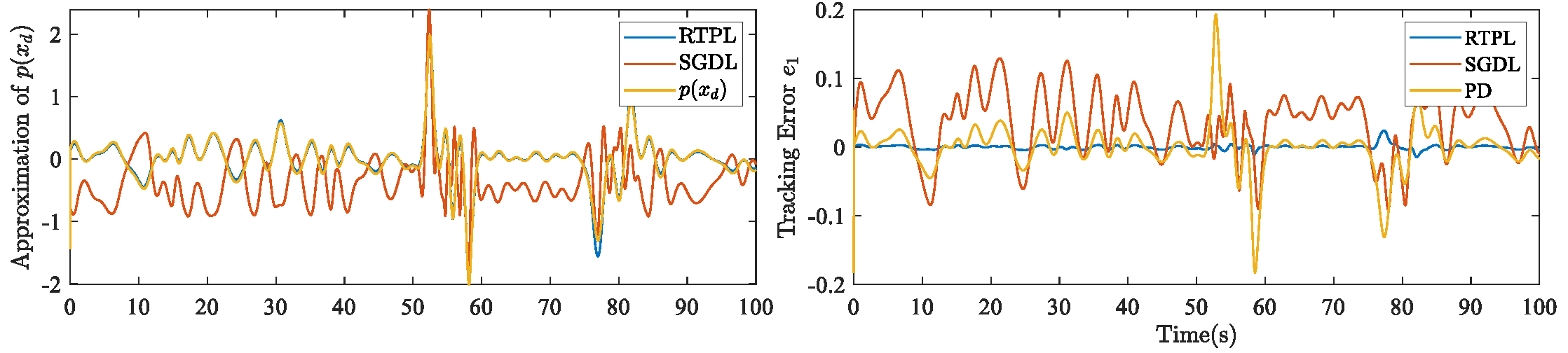}}
	\caption{Tracking and approximation performance of the learned knowledge along $\varphi_B$}
	\label{fig10}\vspace{-1em}
\end{figure*}

\subsection{Learning Control with Abrupt Parameter Perturbation}\label{section6-3}
To validate the learning and control performance of RTPL and SGDL under abrupt parameter perturbation, the same learning control task as in Section \ref{section6-1} is considered. It should be noted that the following parameter perturbation happens in the system \eqref{eq78}  
\begin{equation}\label{eq82}
l = \left\{ \begin{gathered}
	0.2,\quad for\quad 0 \leqslant t < 50s \hfill \\
	0.8,\quad for\quad 50s \leqslant t \leqslant 100s. \hfill \\ 
\end{gathered}  \right.
\end{equation}
The hyperparameter setting $a)$ in TABLE \ref{table1} is adopted. After the learning process, the same knowledge extraction mentioned in Section \ref{section6-1} is adopted. Fig. \ref{fig11} and Fig. \ref{fig12} show the learning and reuse phase performance respectively. While both RTPL and SGDL demonstrate resistance to abrupt parameter perturbation, RTPL maintains higher learning accuracy and speed than SGDL. Interestingly, theoretical analysis suggests that SMRLS can achieve re-approximation of an unknown function along a repetitive trajectory within one period after parameter perturbation \cite{r25} and the results in Fig. \ref{fig11} (c) indicate that RTPL also possesses this advantage. 

\begin{figure*}[htbp]
	\makebox[\textwidth][c]{\includegraphics[height=0.2\textwidth,width=1.05\textwidth]{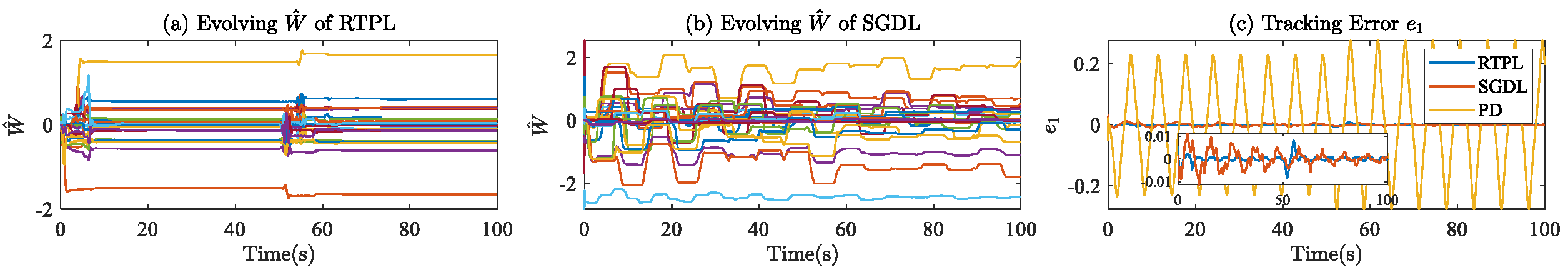}}
	\caption{Learning phase performance along $\varphi_A$ with abrupt parameter perturbation at $50s$}
	\label{fig11}\vspace{-1em}
\end{figure*}

\begin{figure*}[htbp]
	\makebox[\textwidth][c]{\includegraphics[height=0.2\textwidth,width=1.05\textwidth]{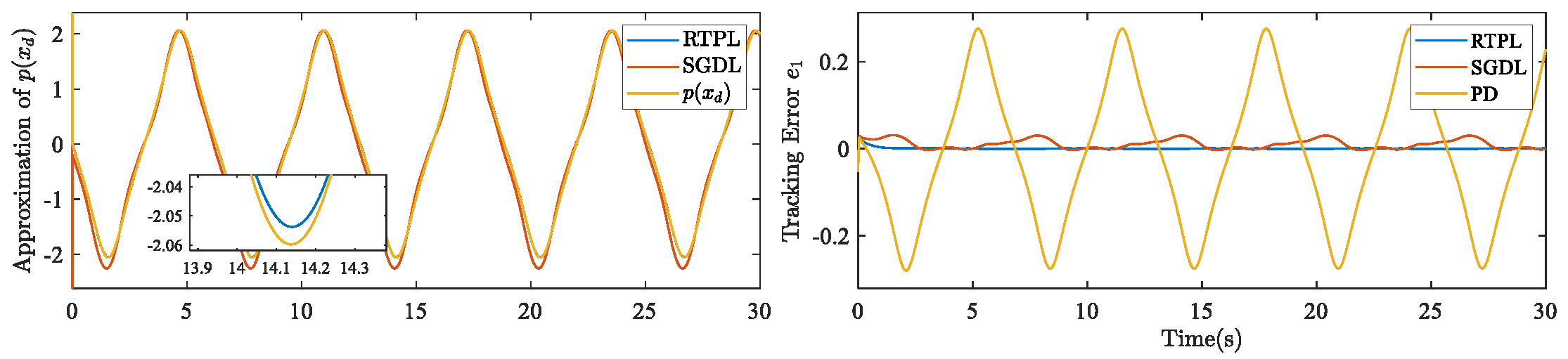}}
	\caption{Reuse phase performance along $\varphi_A$ with $l=0.8$}
	\label{fig12}\vspace{-1em}
\end{figure*}

\subsection{Generalization Capability and Knowledge Accumulation}\label{section6-4}
To validate the generalization capability of RTPL, a random NURBS reference trajectory ${\varphi _C}$ with ${x_{d1}} \in \left[ { - 1,1} \right]$ is adopted in the learning phase and shown in Fig. \ref{fig13}. The angular velocity $x_{d2}$ is also normalized into $\left[ { - 1,1} \right]$ for RTPL. The hyperparameter setting $a)$ in TABLE \ref{table1} is adopted. After the learning phase, the same knowledge extraction method mentioned in Section \ref{section6-1} is adopted and the recorded weights are used to design feedforward controllers for the tracking control task along another trajectory $\varphi_D$ as follows  
\begin{equation}\label{eq83}
{x_{d1}}(t) = \frac{{(20 + t)\sin t}}{{120}},t \in \left[ {0,100s} \right]. 
\end{equation}

\begin{figure}[htbp]
	\centering
	\includegraphics[scale=0.6]{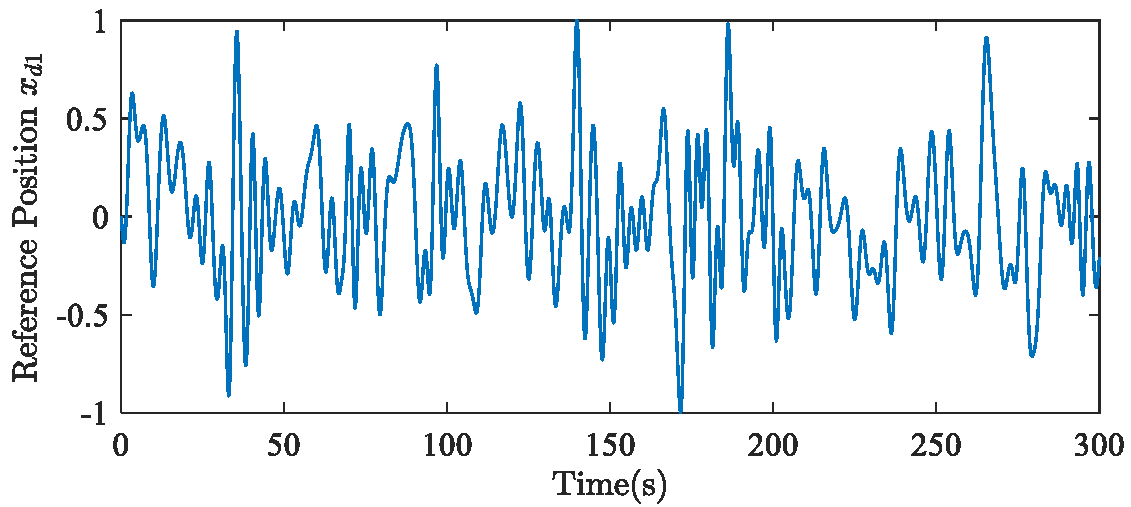}
	\caption{Random NURBS trajectory $\varphi_C$}\vspace{-1em}\label{fig13}
\end{figure}

\begin{figure*}[htbp]
	\makebox[\textwidth][c]{\includegraphics[height=0.2\textwidth,width=1.05\textwidth]{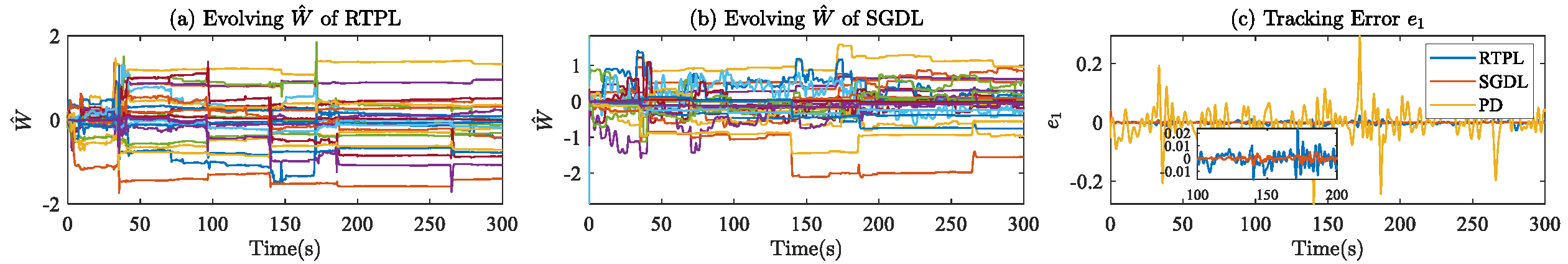}}
	\caption{Learning phase performance along $\varphi_C$}
	\label{fig14}\vspace{-1em}
\end{figure*}

Fig. \ref{fig15} shows reuse phase performance of RTPL and SGDL along $\varphi_D$. After the training of $300s$, the knowledge acquired by RTPL and SGDL has effectively enhanced the tracking accuracy of the control task along $\varphi_D$ while the knowledge acquired by RTPL exhibits stronger generalization capability than SGDL. 

\begin{figure*}[htbp]
	\makebox[\textwidth][c]{\includegraphics[height=0.2\textwidth,width=1.05\textwidth]{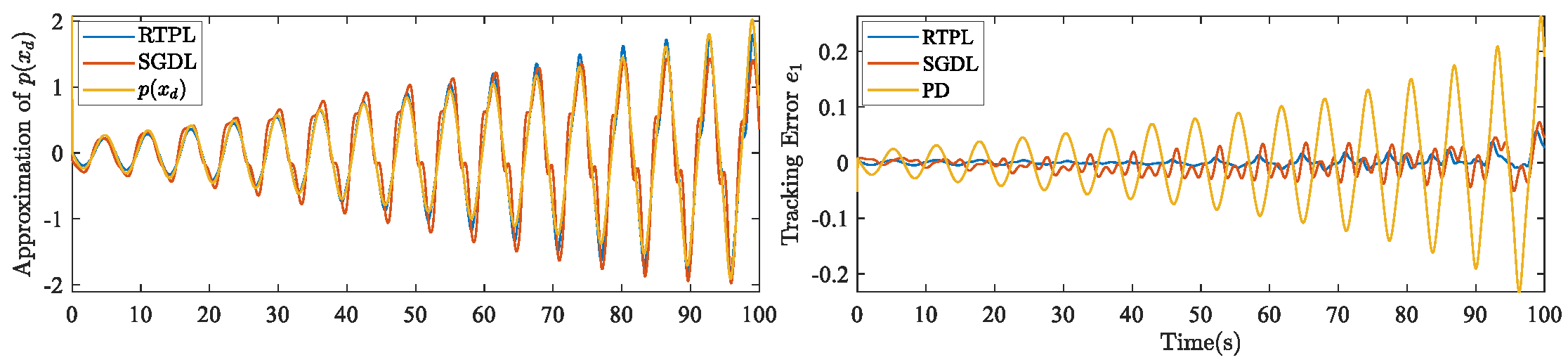}}
	\caption{The performance of the knowledge obtained from $\varphi_C$ in the tracking control task along $\varphi_D$}
	\label{fig15}\vspace{-1em}
\end{figure*}

In order to investigate the relationship between the accuracy of the learned knowledge and the learning duration, the weight vector is recorded every $30s$ during the learning process along $\varphi_C$. Then the ten groups of the weight vector are used for the tracking control task along the test trajectory $\varphi_D$. To measure the accuracy of the learned knowledge, the following integrated squared error (ISE) of $e_1$ and $\tilde{p}(x_d)=p(x_d)-\hat{p}(x_d)$ along $\varphi_D$ is considered  

\begin{equation}\label{eq84}
	\begin{small}
\left\{ \begin{aligned}
	&E({e_1}) = \int_0^{{T_D}} {e_1^2d\tau }  \hfill \\
	&E(\tilde p({x_d})) = \int_0^{{T_D}} {{{\tilde p}^2}({x_d})d\tau }  \hfill \\ 
\end{aligned}  \right.
\end{small}
\end{equation}
where $T_D=100s$ is the duration of $\varphi_D$. The relationship between the ISE and the learning duration is shown in Fig. \ref{fig16}. In general, the accuracy of the learned knowledge is improved as the training time increases which demonstrates the progressive characteristic of the RTPL method. Compared with SGDL, which exhibits fluctuations in learning accuracy with increasing learning duration, RTPL demonstrates progressively improving accuracy on along $\varphi_D$ as the learning time increases. Therefore, the results indicate that RTPL is capable of accumulating knowledge progressively in long-term real-time learning tasks. 

\begin{figure}
	\centering
	\includegraphics[height=0.32\textwidth,width=0.43\textwidth]{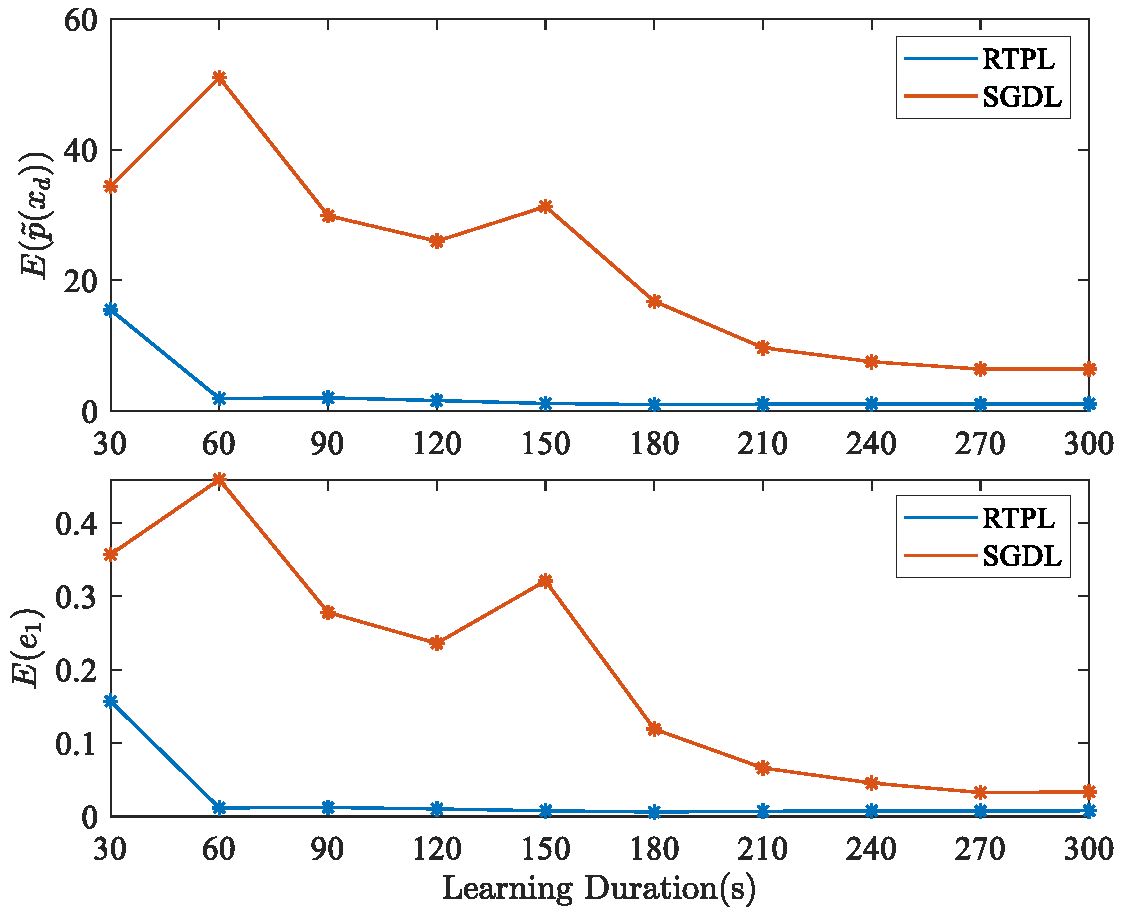}
	\caption{Increasing accuracy of the knowledge}\vspace{-1em}\label{fig16}
\end{figure}

\section{Conclusion}\label{section7}
In this paper, an RBFNN based filter-free learning control scheme named RTPL is proposed by applying SMRLS according to measurable errors in closed-loop systems. Theoretical analysis shows that the exponential convergence of weights can be realized along arbitrary reference trajectories. Thanks to the unique memory mechanism of SMRLS, the passive knowledge forgetting phenomenon is suppressed and RTPL achieves multiple merits including fast convergence speed, robustness to hyperparameter setting, good generalization ability and resistance to parameter perturbation. More interestingly, RTPL can progressively accumulate knowledge from various control tasks such that an increasingly accurate approximation of the unknown dynamics is obtained. It is believed that RTPL will have promising applications in many fields including motion control, robotics, human motor control, etc, which will be further investigated. 

%%%%%%%%%%%%%%%%%%%%%%%%%%%%%%%%%%%%%%%%%%%%%%%%%%%%%%%

\ifCLASSOPTIONcaptionsoff
  \newpage
\fi

\small
\bibliographystyle{ieeetr}
\bibliography{refer}

%\iffalse

\begin{IEEEbiography}[{\includegraphics[width=1in,height=1.25in,clip,keepaspectratio]{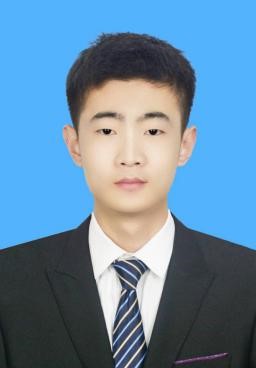}}]{Yiming Fei}
	received the B.Eng. and M.Eng. degrees from the Harbin Institute of Technology, China, in 2020 and 2023, respectively. He is currently pursuing a Ph.D. degree in Computer Technology at Zhejiang University, China. His current research interests include reinforcement learning, cognitive neuroscience, neural network control, system identification and motion control.
\end{IEEEbiography}
\begin{IEEEbiography}[{\includegraphics[width=1in,height=1.25in,clip,keepaspectratio]{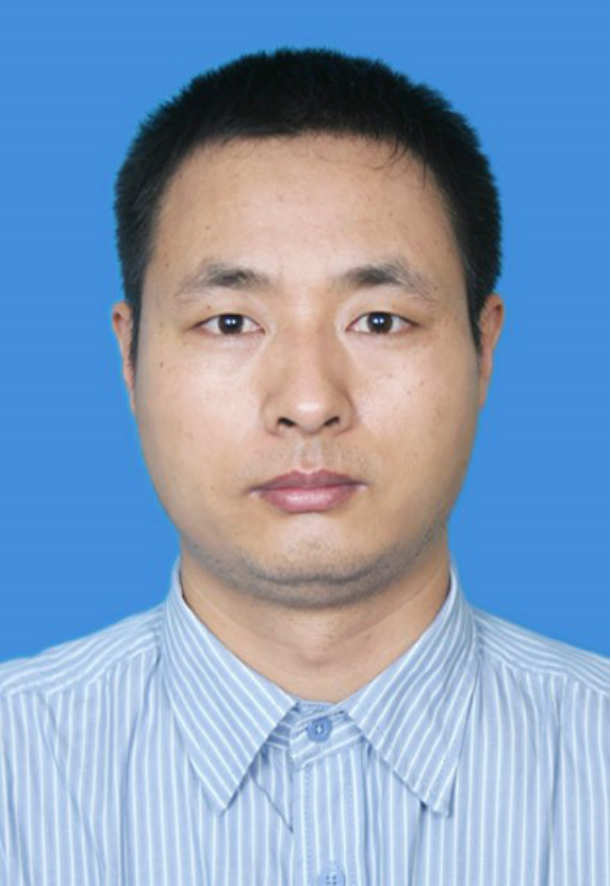}}]{Jiangang Li}
	(M09-SM20) received the B.Eng., M.Eng., Ph.D. degrees from the Xi’an Jiaotong University, China, in 1999, 2002 and 2005, respectively. From 2007 to  present, he has been an Associate Professor in control science and engineering with the School of Mechanical Engineering and Automation, Harbin Institute of Technology Shenzhen, China.	From 2015 to 2016, he has been a Visiting Associate in computing and mathematical sciences with the California Institute of Technology. His general research interests include high velocity and high performance control system design, motion control and motion planning.
\end{IEEEbiography}

\begin{IEEEbiography}[{\includegraphics[width=1in,height=1.25in,clip,keepaspectratio]{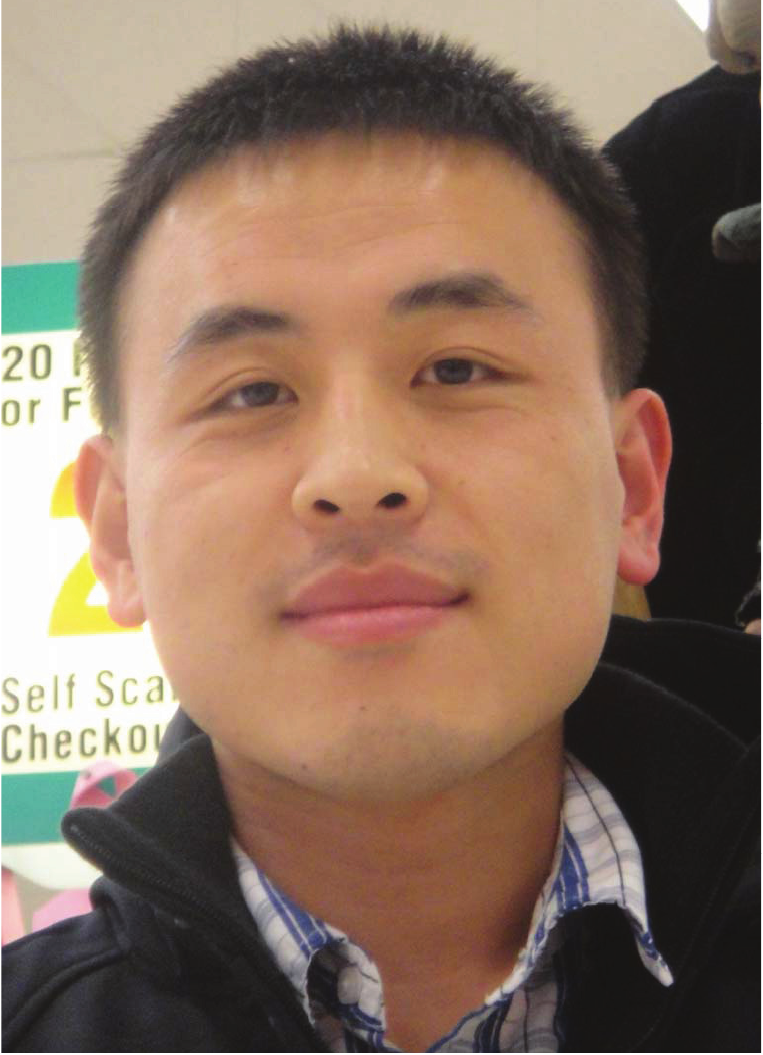}}]{Yanan Li}
	(M14-SM21) received the B.Eng. and M.Eng. degrees from the Harbin Institute of Technology, China, in 2006 and 2008, respectively, and the Ph.D. degree from the National University of Singapore, Singapore, in 2013. He is currently a Reader in Robotics with the Department of Engineering and Design, University of Sussex, Sussex, U.K. From 2015 to 2017, he was a Research Associate with the Department of Bioengineering, Imperial College London, U.K. From 2013 to 2015, he was a Research Scientist with the Institute for Infocomm Research, Agency for Science, Technology and Research, Singapore. His general research interests include human-robot interaction, robot control and control theory and applications.
\end{IEEEbiography}

%\fi


%\end{IEEEbiography}


%\newpage


\end{document}